\documentclass[11pt]{article}
\usepackage{amssymb}\usepackage{amsmath}
\usepackage{amsthm}
\newcommand{\be}{\begin{equation}}
\newcommand{\ee}{\end{equation}}

\newcommand{\tr}{\mathop{\rm tr}\nolimits}

\def\ST{{\mathsf T}}
\def\SP{{\mathsf P}}

\def\ot{\otimes}

\def\la{\lambda}

\def\CT{{\mathcal T}}

\def\CE{{\mathcal E}}

\newcommand{\scp}[2]{\langle #1, #2 \rangle}
\theoremstyle{plain}
\newtheorem{defn}{\bfseries Definition} 
\newtheorem{thm}{\bfseries Theorem}
\newtheorem{propn}{\bfseries Proposition}
\newtheorem{lem}{\bfseries Lemma}
\newtheorem{cor}{\bfseries Corollary}
\theoremstyle{remark}
\newtheorem{rem}{\sc Remark} 
\newtheorem{ex}{\sc Example}
\begin{document}

\begin{center}
{\bf\large\sc 
Tensor space representations of Temperley--Lieb algebra \\[1mm]
 via orthogonal projections of rank $r \geq 1$ 
} 

\vspace*{4mm}
{ Andrei Bytsko ${}^{1,2}$  }

\vspace*{3mm}

{\small
\noindent
${}^{1}$ Department of Mathematics, University of Geneva, 
C.P. 64, 1211 Gen\`eve 4, Switzerland \\ 
${}^{2}$ Steklov Institute of Mathematics,
Russian Academy of Sciences,
Fontanka 27, 191023, St.~Petersburg, Russia}
 
\end{center}
\vspace*{2mm}

\begin{abstract}

Unitary representations of the Temperley--Lieb algebra $TL_N(Q)$
on the  tensor space $({\mathbb C^n})^{\otimes N}$ are considered.
Two criteria are given for determining when an orthogonal 
projection matrix $P$ of a rank $r$ gives rise to 
such a representation. The first of them is the equality of 
traces of certain matrices and the second is the unitary condition for 
a certain partitioned matrix.
Some estimates are obtained on the lower bound of $Q$ 
for a given dimension $n$ and rank~$r$. 
It is also shown that if $4r>n^2$, then $Q$ can take 
only a discrete set of values determined
by the value of~$n^2/r$. In particular, the only allowed value
of $Q$ for $n=r=2$ is $Q=\sqrt{2}$.
 Finally, properties of the Clebsch--Gordan coefficients of the 
quantum Hopf algebra $U_q(su_2)$ are used in order to
find all $r=1$ and $r=2$ unitary tensor space representations of $TL_N(Q)$ 
such that $Q$ depends continuously on $q$ and $P$ is
the projection in the tensor square of a simple $U_q(su_2)$
module on the subspace spanned by one or two joint
eigenvectors of the Casimir operator $C$ 
and the generator $K$ of the Cartan subalgebra.  
\end{abstract}

\section{Introduction}

The Temperley--Lieb (TL) algebras play an important role 
in the theory of subfactors, knot theory, and studies of 
discrete models in low dimensional physics.
The Temperley--Lieb  algebra of the type $A_{N-1}$
was introduced in~\cite{TeLi}. 
Recall its definition.

\begin{defn}
Given $Q \in {\mathbb C}$ and an integer $N \geq 2$,
the Temperley--Lieb algebra $TL_N(Q)$ 
is the unital algebra over $\mathbb C$ 
with generators $\ST_1,\,{\ldots}\,,\ST_{N-1}$
and the following defining relations:
\begin{eqnarray}
\label{TL1q}
{}&& \ST_k \, \ST_k = Q \, \ST_k \,, 
 \qquad\qquad\ \  \text{\rm for all}\ \ k  \,, \\
 \label{TL1}
{}&& \ST_k \, \ST_m = 
\ST_m \, \ST_k  \qquad\qquad\  
 \text{\rm for}\ \ |k-m| \geq 2 \,,  \\
\label{TL2}
{}&&  \ST_k \, \ST_{m}  \ST_k = \ST_k \,,  
 \qquad\qquad\,  \text{\rm for} \ \ |k-m|=1  \,.
\end{eqnarray} 
\end{defn}

The Temperley--Lieb algebra has a natural 
linear anti--involution:
\begin{eqnarray}
\label{TL0}
{}&&
 \ST_k^* = \ST_k  
 \qquad\qquad\qquad\quad\!  \text{\rm for all}\ \ k  \,.
\end{eqnarray}

In the present article, we will consider 
a particular class of representations of $TL_N(Q)$
on the tensor product space 
$\bigl({{\mathbb C}^n}\bigr)^{\otimes N}$.
We will denote by $M_n$ the ring of 
$n\,{\times}\,n$ complex matrices,
by $I_n$ the $n\,{\times}\,n$ 
identity matrix, and by $\otimes$ the 
Kronecker product. Given $X \in M_n$, $X^*$ 
will stand for its conjugate transpose.
 
\begin{defn}\label{DefTau}
Given $Q >0 $ and an integer $n \geq 2$,
a homomorphism
$\tau : \text{TL}_N(Q) \to  M_{n^N}$
is a unitary tensor space representation of 
$\text{TL}_N(Q)$ if 
\begin{equation}\label{TP}
   \tau \bigl(\ST_k\bigr) = 
 I_n^{\otimes (k-1)} \otimes T \otimes I_n^{\otimes (N-k-1)}  
 \,, \qquad\quad k=1,\ldots, N{-}1 \,,
\end{equation}
and matrix $T \in M_{n^2}$
satisfies the following relations:
\begin{align*} 
{}& (\mathrm{T1})     &&
  T^* = T ,&&  \\
{}&  (\mathrm{T2})  && 
	 T \, T = Q \, T   , && \\  
{}& (\mathrm{T3})     &&   
 T_{12} \, T_{23} \, T_{12} \, = T_{12}  \,,&&  \\
\label{tl2}
{}& (\mathrm{T4})   &&
 T_{23} \, T_{12} \, T_{23} \, = T_{23}  \,, &&    
\end{align*}
where 
$T_{12} \equiv T \,{\otimes}\, I_n$ and 
$T_{23} \equiv I_n \,{\otimes}\, T$.
\end{defn} 
 
Given  $T \,{\in}\, M_{n^2}$ 
satisfying (T2)--(T4) with $Q\,{>}\,0$,
set $R=q\,I_{n^2}-T$, where $q$ is a root of the 
equation $q + q^{-1} =Q$. Note that $R$ is
invertible, $R^{-1}=q^{-1} I_{n^2}-T$.
Consider the following map 
\hbox{${\mathcal R} : {\mathbb C} \to M_{n^2}$}:
\begin{equation}\label{RT2}
  {\mathcal R}(u) = \begin{cases}
  u\, R -R^{-1} \qquad\qquad \text{if}\ \,
  Q \neq 2 \,,  \\
  u\, R + I_{n^2}  , 
  \qquad\qquad \text{if}\ \, Q=2 \,.
  \end{cases}
\end{equation}

One of the main motivations to study
tensor space representations of the TL algebra
is  the following well--known fact: the map (\ref{RT2})
provides a non--trivial example of an {\em R--matrix}, i.e.
a solution to the {\em Yang--Baxter equation}:
\begin{equation}\label{YB}
  {\mathcal R}_{12}(u)\, {\mathcal R}_{23}(u \bullet v) \,
   {\mathcal R}_{12}(v) = {\mathcal R}_{23}(v) \, 
   {\mathcal R}_{12}(u \bullet v) \, {\mathcal R}_{23}(u)\,,
\end{equation} 
where $\bullet$ stands for summation if $Q=2$ and
for multiplication otherwise.

In its turn, an R--matrix is the cornerstone for 
building quantum  integrable models known as 
{\em spin chains}, see, e.g.~\cite{Fad1}. 
{}From this perspective, the most interesting tensor space
representations of $TL_N(Q)$ are those with varying $Q$,
i.e. such representations where $T$ depends on some
parameters and $Q$  varies within a certain range 
when the parameters change. Indeed, such a representation
allows us to construct a parametric family of R--matrices and,
therefore, of integrable models.

\vspace{1mm}{\small
\begin{ex}\label{ex1}
The most known example of such a type is given by 
\begin{equation}\label{xxz1}
 T(q;\zeta) =   
 \begin{pmatrix}
  0 & 0 & 0 & 0 \\
  0 & q & \zeta & 0 \\
  0 & \zeta^{-1} & q^{-1} & 0 \\
  0 & 0 & 0 & 0
 \end{pmatrix} 
 ,  \qquad\qquad  Q=q+q^{-1} .
\end{equation}
For $q>0$ and $|\zeta|=1$, this $T(q;\zeta)$ defines 
a unitary tensor space representation of $TL_N(Q)$. 
For $q \neq 1$, the spin chain corresponding
to $T(q; \pm 1)$ is called the XXZ model and, for
$q=1$, it is the Heisenberg spin chain.
\end{ex}
}

\vspace{1mm}
\begin{rem} \label{QEXT2}
Let a family of unitary tensor space representation of $TL_N(Q)$
be defined by $T(q)$ which satisfies (T1)--(T4)
when $q$ varies continuously within a subset 
$S \subset {\mathbb C}$. If entries of $T(q)$ are rational 
functions in~$q$ with poles contained in
a subset $\tilde{S}$, then equations (T2)--(T4) imply that certain 
functions rational in $q$ vanish on~$S$. 
But then these functions must vanish identically.
Therefore, for
$q \in {\mathbb C} \setminus (S \cup \tilde{S})$, $T(q)$
will not be Hermitian but it will satisfy (T2)--(T4) 
(where $Q$ may not be real) and thus will define a family of
non--unitary tensor space representation of $TL_N(Q)$. 
\end{rem}

\vspace*{1mm}{\small 
\begin{ex}\label{QEXT}
 For $q \in {\mathbb C} \setminus {\mathbb R}$ and
$\zeta \in {\mathbb C} \setminus \{0\}$,
$T(q;\zeta)$ given by (\ref{xxz1}) is not Hermitian but 
it satisfies relations (T2)--(T4).
\end{ex}
}

\vspace{1mm}
\begin{rem} 
Unitary tensor space representations of $TL_N(Q)$
that for some values of parameters extend to
non--unitary ones can be used to construct non--Hermitian
operators with real spectrum. 
For instance, if $T(q;\zeta)$ is given by (\ref{xxz1}), then
$H= T_{12}(q;\zeta) - T_{23}(q;\zeta)$ is Hermitian 
only for real $q$ but its spectrum remains real
also for $q=e^{i \gamma}$, $\gamma \in \mathbb R$
providing that $4 \cos^2 \gamma \geq 1$. 
See \cite{By1} for further examples of such a type.
\end{rem}
\vspace{1mm}

The three important characteristics of a 
tensor space representation of $TL_N(Q)$ are the value of $Q$, 
the dimension $n$ which determines the size of $T$, and
the rank of $T$, $r \equiv \text{rank} (T)$.
In what follows, somewhat abusing the terminology,
we will refer to $r$ as simply the rank of 
a representation.

In the rank one case, properties of spin chains based on 
the TL R--matrices (\ref{RT2}), in particular, 
spectra of the TL Hamiltonians $H=T_{12}+T_{23}+\ldots$ have
been studied by a large number of authors, see e.g.
\cite{AuKl,BaBa,GiPy,Kul,MaSa}. These studies
used mainly the representation determined by 
$T(q,-1)$ or its higher spin analogue (cf. Section~\ref{UQ}). 
This representation enjoys the great popularity because on the one hand
it is a  representation with varying $Q$ and
thus it can be used to study parametric families of Hamiltonians and
on the other hand it is related to the quantum Hopf algebra $U_q(su_2)$
(cf.~Section~4).

In the higher rank case, $r \geq 2$, some tensor space 
representations of $TL_N(Q)$ were constructed in \cite{Kul3,WX}
for $r=n \geq 2$ but they correspond only 
to a specific value of $Q$, namely $Q=\sqrt{n}$.
And, to the best of author's knowledge, spin chains 
based on higher rank tensor space representations
of $TL_N(Q)$ have not yet been studied.

The goal of the present article is to consider 
certain problems related to construction of 
unitary tensor space representations of $TL_N(Q)$
of an arbitrary rank~$r$. 

The paper is organized as follows. In Section~2,
we first comment on a certain redundancy in 
equations (T1)--(T4). Then we give two criteria
for determining when an orthogonal 
projection matrix $P$ gives rise to a 
unitary tensor space representation of $TL_N(Q)$
(matrix $T$ in (T1)--(T4) is always a scalar
multiple of a projection matrix). The first
of them is the equality of traces of certain matrices 
and the second is the unitary condition for 
a certain partitioned matrix.
In Section~3, we give some estimates on the 
lower bound of $Q$ if the dimension $n$ and
the rank $r$ are given. In particular,
we show that $Q \geq n/r$ and that this
yields the sharp lower bound if $r=1$.
Using the Jones--Wenzl projector, we show 
that if $r>n^2/4$, then $Q$ can take
only a discrete set of values determined
by the value of~$n^2/r$. It follows, in 
particular, that the only allowed value
of $Q$ for $n=r=2$ is $Q=\sqrt{2}$.
At the end of the section, the estimates on
$Q$ are sharpened for some special cases. 
In Section~4, we use the results of Section~2
as well as some properties of the 
Clebsch--Gordan coefficients of the quantum Hopf algebra
$U_q(su_2)$ (for a generic positive $q$ and for $q=1$) 
in order to find all varying $Q$ unitary 
tensor space representations of $TL_N(Q)$ 
of rank one and rank two where $P$ is the orthogonal projection 
in the tensor product of two spin $S$ representations 
of~$U_q(su_2)$ on a subspace spanned by, respectively, 
one or two joint eigenvectors of $C$ and $K$ (the Casimir 
operator and the generator of the Cartan subalgebra).
 The proofs of all statements are given in 
the Appendix.

All varying $Q$ rank two representations found in Section~4
correspond to $S=1$, that is $n=3$. In the subsequent
article \cite{By2}, we will construct families
of varying $Q$ rank two unitary tensor space representations 
of $TL_N(Q)$ for $n=3k$ and $n=4k$, $k \in {\mathbb N}$
and will also give a complete classification of
representations of rank one.

\section{Criteria for an orthogonal projection}

\subsection{Remarks on equations (T1)--(T4)}

We commence with the following simple remark. 
If $T, T' \in M_{n^2}$ are solutions to (T1)--(T4) 
corresponding to the same value of $Q$ and
$\text{rank}(T)=\text{rank}(T')$,
then, by the spectral theorem, 
these matrices are unitarily similar, $T'=GTG^*$,
where $G$ is unitary.
But the converse is not true: if $T \in M_{n^2}$ is 
a solution to (T1)--(T4) and $G\in U(n^2)$, then 
$T'=GTG^*$ is not in general a solution to (T3)--(T4). 
However, if $G=g\,{\otimes}\,g$, $g \in U(n)$, 
then the unitary similarity 
transformation
\begin{equation}\label{TGG}
  T' = \bigl(g \otimes g \bigr) \, T \,
   \bigl(g^* \otimes g^* \bigr) \,,
\end{equation}
does send a solution $T$ to equations (T1)--(T4)
to another solution, $T'$, to these equations.
Clearly, $T$ and $T'$ related as in (\ref{TGG}) have 
equal ranks and correspond to the same value of~$Q$.
It is thus natural to study solutions to
(T1)--(T4) up to the unitary equivalence~(\ref{TGG}). 

Next, let us remark that relations (T1)--(T4) 
in the definition of a tensor space 
representation are somewhat redundant.
 
Here and below $\tr$ denotes the 
standard matrix trace. 

\begin{propn}\label{TTQT}
a) If $T \in {M}_{n^2}$
satisfies relations (T3) and (T4),
then $T^2=Q\, T$, where $Q \in {\mathbb C}\setminus\{0\}$
if\, $\tr T \neq 0$ and $Q=0$ if\, $\tr T = 0$.
 \\[1mm]
b) If $T \in {M}_{n^2}$ 
satisfies relation (T1) and 
any two of the three relations (T2)--(T4),
then $T$ satisfies all the relations (T1)--(T4).
\end{propn}

{\small
\begin{ex}
For $T(q;\zeta)$ given by (\ref{xxz1}), 
we have  $\tr T(q;\zeta) = (q+q^{-1})$.
$T(q;\zeta)$~is a scalar multiple of a rank one 
projection if $q^2 \neq -1$ and  
it is a nilpotent of order two if $q^2 = -1$. 
(The spin chain corresponding to the nilpotent case 
is known as the XX0 or XX model.)
\end{ex}
} 

\subsection{Trace conditions}

Every $T \in M_{n^2}$ that satisfies relations 
$\mathrm{(T1)}$ and 
$\mathrm{(T2)}$ is a scalar multiple
of an orthogonal projection, i.e.
$T=Q P$, where $P \in M_{n^2}$, $P=P^2=P^*$.
Without a loss of generality, we always assume that
$Q>0$ (because a negative $Q$ can be made positive
by the trivial transformation $T \to -T$).

If the rank of~$P$ is $r$, then $\tr\, T = Q\,r$. Furthermore, 
\begin{equation}\label{tr123}
\tr_{123}\,(T_{12}) = \tr_{123}\,(T_{23}) = Q\,r\,n \,,
\end{equation}
where $\tr_{123}$ stands for the matrix trace in
$M_{n^3}$.

The problem of constructing unitary
tensor space representations of $TL_N(Q)$,
that is finding solutions $T$ to
equations (T1)--(T4), amounts to finding 
suitable orthogonal projections in 
$({\mathbb C}^n)^{\otimes 2}$.
Remarkably, such projections can be characterised
by just a single scalar condition.
 
\begin{thm}\label{trTTb}
Let $P \in {M}_{n^2}$
be an orthogonal projection of rank~$r \geq 1$
and suppose that $P_{12} P_{23} \ne 0$.
Then a solution to (T3)--(T4) of the form
$T=QP$, where $Q>0$, exists if and only if
the following equality holds:
\begin{equation}\label{trP} 
   \bigl( \tr_{123}\,(P_{12} P_{23}) \bigr)^2 =  
   n\,r \, \tr_{123}\, (P_{12} P_{23} )^2  \,.
\end{equation}
If equality (\ref{trP}) holds, then 
relations (T1)--(T4) are satisfied for $T=QP$, where
\begin{equation}\label{QP} 
  Q^2= \frac{n\,r}{\tr_{123}\,(P_{12} P_{23}) }  \,.
\end{equation}
\end{thm}

As a consequence, matrix equations (T3)--(T4) 
in the definition of a unitary tensor space 
representation of $TL_N(Q)$ can be
replaced by scalar equations as follows.

\begin{propn}\label{trTT}
Suppose that $T \in {M}_{n^2}$ has rank $r$ 
and satisfies relations (T1) and (T2) with $Q>0$.
Then the following statements are equivalent:\\[0.5mm]
a)  $T$ satisfies  relations (T3)--(T4).\\[0.5mm]
b) $\tr_{123}\,(T_{12} T_{23}) \neq 0$ and
the following equality holds:
\begin{equation}\label{trtttt} 
   \bigl( \tr_{123}\,(T_{12} T_{23}) \bigr)^2 =  
   n\,r \, \tr_{123}\, (T_{12} T_{23} )^2  \,.
\end{equation}
c) The following equalities hold:
\begin{equation}\label{trt} 
   \tr_{123}\,(T_{12} T_{23}) = n\,r \,, \qquad
   \tr_{123}\, (T_{12} T_{23} )^2 = n\,r \,.
\end{equation}
\end{propn}

Theorem~\ref{trTTb} can be used, in particular,
in order to search for solutions to (T1)--(T4)
numerically. For this purpose, we have to
choose some orthonormal basis 
$\{ y_a \}_{a=1}^{n^2}$ of 
$\bigl({\mathbb C}^n\bigr)^{\otimes 2}$
and then test condition (\ref{trP}) for all
projections of the form 
$P=\sum_{a=1}^{n^2} \varepsilon_a P_a$,
where $\varepsilon_a$ is 0 or 1 and 
$P_a$ is the projection on the one--dimensional
subspace spanned by~$y_a$.
Moreover, if we are interested in
representations with varying $Q$,
it suffices to check only the cases where
$\sum_{a=1}^{n^2} \varepsilon_a \leq n^2/4$
(see Section~3).

\subsection{Unitarity condition}

Another way to characterize an orthogonal
projection $P$ which gives rise to a unitary
tensor space representation of $TL_N(Q)$ is
to find a condition on the subspace on which
$P$ projects.

Let $\scp{\,}{}$ denote the standard inner product 
on~${\mathbb C}^n$ and let ${\CE}=\{e_a\}_{a=1}^n$ be 
a  basis of ${\mathbb C}^n$ orthonormal w.r.t. $\scp{\,}{}$.
Then a vector $v \in {\mathbb C}^n \otimes {\mathbb C}^n$ 
is determined by 
the matrix $V$ of its coefficients,
$v  = \sum_{a,b=1}^n V_{ab}\, e_a \otimes e_b$.
Under a unitary change of the basis, 
$e_a = \sum_{b=1}^n {g_{ab} e'_b}$, 
$g \in U(n)$, the matrix of coefficients  
transforms as follows:
\begin{equation}\label{gVg}
 V' = g^t \, V \, g .
\end{equation}
Here and below we use the following notations for
matrix operations: $\bar{X}$, $X^t$, and $X^*$
stand, respectively, for the complex conjugate, 
the transpose, and the conjugate transpose 
of a matrix~$X$. 

Given an $r$--dimensional vector subspace $\CT \subset 
 {{\mathbb C}^n} \,{\otimes}\, {{\mathbb C}^n}$, 
we will write  ${\CT} \sim  \{V_1,\ldots,V_r \}$ 
if $V_1,\ldots,V_r$ are the matrices of coefficients
of an orthonormal set of vectors $v_1,\,{\ldots}\,,v_r$
which is a spanning set of~$\CT$.
The orthonormality condition implies that
\begin{equation}\label{vv}
 \scp{v_s}{v_m} = \tr \bigl( V_s^*  V_m \bigr) 
  =\delta_{sm} \,.
\end{equation}
The  orthogonal projection onto  
$\CT$ is given by
$P_\CT = \sum_{s=1}^r v_s \scp{v_s}{\cdot\,}$.
In the basis $\CE$,  
the operator $e_a \scp{e_b}{\cdot\,}$ 
is represented by the matrix
$E_{ab} \in {M}_{n}$ such that
$\bigl(E_{ab}\bigr)_{ij} = \delta_{ai} \delta_{bj}$.
Therefore, the projection $P_\CT$ is represented 
by the following matrix: 
\begin{equation}\label{PiTau}
  P_\CT = 
 \sum_{s=1}^r \sum_{a,b,c,d=1}^n (V_s)_{ab} \, 
 (\bar{V}_s)_{cd} \  E_{ac} \otimes E_{bd} \,,
\end{equation}
where $\otimes$ stands for the Kronecker product.

Each $T \in M_{n^2}$ which satisfies (T1)--(T2) 
and has rank $r$ is determined by a set of matrices 
$\{V_1,\ldots,V_r \}$ such that $T=Q\,P_{\CT}$,
where ${\CT} \sim  \{V_1,\ldots,V_r \}$.  

\vspace*{1mm}{\small
\begin{ex}
For $T(q;\zeta)$ given by (\ref{xxz1}),
 we have $T(q;\zeta)=(q+q^{-1})P_\CT$ 
with $\CT \sim \{V\}$,
where 
\begin{equation}\label{vTq}
 V = \frac{1}{\sqrt{q^2+1}} 
 \biggl( \begin{matrix}
  0 & \zeta\, q \\
  1  & 0 
 \end{matrix} \biggr) \,, \qquad q>0 \,, \quad |\zeta|=1\,.
\end{equation} 
\end{ex}
}
\vspace*{1mm}

Given an
$r$--dimensional subspace $\CT \subset 
 {{\mathbb C}^n} \,{\otimes}\, {{\mathbb C}^n} $,
${\CT} \sim  \{V_1,\ldots,V_r \}$,  
let us introduce the following partitioned matrix
consisting of $r^2$ blocks  
of the size $n \,{\times}\, n$:
\begin{align}\label{WV} 
  W_{\CT}  =   \left( 
 \begin{matrix}
   V_1 \bar{V}_1 & V_2 \bar{V}_1 & \ldots \\
     V_1 \bar{V}_2 & V_2 \bar{V}_2 & \ldots \\
     \vdots & \vdots & \ddots 
     \end{matrix} \right)
  = \sum_{s,m=1}^r  E_{sm} \otimes V_m \bar{V}_s \,.
\end{align}

By Theorem~\ref{trTTb}, finding a solution to
(T1)--(T4)
amounts to finding a subspace $\CT$ such that
the corresponding orthogonal projection $P_{\CT}$ 
satisfies relation~(\ref{trP}).
Let us reformulate relation (\ref{trP}) 
as a condition on the matrix~$W_{\CT}$.

\begin{thm}\label{PROPTLW}
Let $P_{\CT}$ be the orthogonal projection onto
an $r$--dimensional subspace 
$\CT\subset {{\mathbb C}^n} \,{\otimes}\, {{\mathbb C}^n} $,
${\CT} \sim  \{V_1,\ldots,V_r \}$ and let $W_\CT$ 
be the corresponding matrix defined in~(\ref{WV}).  
Then $T=Q P_\CT$, $Q>0$ is a solution to
(T1)--(T4) if and only if $Q W_\CT$ is a unitary matrix,
\begin{equation}\label{TLP}
  Q\, W_\CT \in U(nr) \,.
\end{equation}  
\end{thm}

Since unitarity of $Q\, W_\CT$ implies
unitarity of $Q\, \bar{W}_\CT$, 
$Q\, W^t_\CT$, and $Q\, W^*_\CT$, 
we deduce the following.

\begin{cor}\label{BTH} 
If $T=Q P_\CT$, $\CT \sim \{V_1,\ldots,V_r\}$
is a solution to (T1)--(T4), then so are
$T'=Q P_{\CT'}$, $T''=Q P_{\CT''}$, and $T'''=Q P_{\CT'''}$,
where $\CT' \sim \{\bar{V}_1,\ldots,\bar{V}_r\}$,
$\CT'' \sim \{V^t_1,\ldots,V^t_r\}$, and
$\CT''' \sim \{V^*_1,\ldots,V^*_r\}$.
\end{cor}

\begin{rem}\label{HGW}
The validity of condition (\ref{TLP}) depends  neither 
on a particular choice of the orthonormal spanning set of
$\CT$ nor on a particular choice of the orthonormal basis 
$\CE$ of ${{\mathbb C}^n}$.
Indeed, for two different orthonormal spanning sets,
 $\{v_s\}_{s=1}^r$ and
$\{v'_s\}_{s=1}^r$, where 
$v'_s = \sum_{k=1}^r {h_{sk} v_k}$,
$h \in U(r)$, the corresponding $W$ matrices
are related by a unitary transformation, namely
$W'_{\CT} = 
(\bar{h} \,{\ot}\, I_n) \, W_{\CT} \, (h^t \,{\ot}\, I_n)$.
For two different orthonormal bases of ${{\mathbb C}^n}$, 
the matrices of coefficients are related
as in (\ref{gVg}) and so the $W$ matrices corresponding
to the same subspace $\CT$ are also related
by a unitary transformation, namely
$W'_{\CT} = (I_r \ot g^t) \, W_{\CT} \, (I_r \ot \bar{g})$.
In either case, the unitarity of $W_\CT$ implies
the unitarity of $W'_\CT$. 
\end{rem}

\begin{rem}
Condition (\ref{TLP}) admits also the following 
formulation. Let $J$ be the unitary involutive automorphism
of ${\mathbb C}^n\,{\otimes}\,{\mathbb C}^n$ 
which maps a vector $v$ with the coefficient matrix $V$
into the vector $J(v)$ with the coefficient matrix
$V^t$. Note that, by (\ref{gVg}), 
the map $v \to J(v)$ is independent 
of a choice of the basis $\CE$ of ${\mathbb C}^n$. 
Observe that $(W_{\CT})^*= W_{J(\CT)}$.
Therefore, condition (\ref{TLP}) is equivalent
to the requirement that $W_{J(\CT)}$ is a scalar
multiple of the inverse to $W_{\CT}$.
\end{rem}

\section{On the range of $Q$}

An interesting problem is to determine the range
of possible values of $Q$ in (T2) for
solutions to (T1)--(T4) if the rank $r$
and the dimension $n$ of the underling 
space ${\mathbb C}^n$ are given.

\vspace{1mm}{\small
\begin{ex}
For $T(q;\zeta)$ given by (\ref{xxz1}) with
$q >0$, we have $Q=q+q^{-1} \in [2,+\infty)$.
As we will see below, $Q=2$ is the sharp lower
bound in the $n=2$, $r=1$ case. 
\end{ex} 
}

\subsection{Rank one case}\label{ROC}

For $r=1$, the normalization condition (\ref{vv})
and the unitarity condition (\ref{TLP}) acquire the
following form: 
\begin{equation}\label{TLVVr}  
   \tr(V\,V^*) =1 \, \qquad\qquad  
   V\, \bar{V} \, V^t V^* 
 =  Q^{-2}   \, I_n \,.
\end{equation}
Clearly, $V$ must be non--singular. 
Taking this into account, we derive from (\ref{TLVVr}) the 
following expressions for $Q$:
\begin{equation}\label{muV} 
  Q^2 =  \bigl|\det V \bigr|^{-\frac{4}{n}} 
\,, \qquad\qquad
 Q^2 =  \tr\bigl( ({V}^* V)^{-1}\bigr)  \,.
\end{equation}
They, in turn, allow us to find the lower bound for $Q$
in the rank one case.

\begin{propn}\label{QVN}
Suppose that $V \in {M}_{n}$
 satisfies relations (\ref{TLVVr}). Then\\[1mm]
a) The following inequality holds:
\begin{equation}\label{Qdim} 
 Q^2 \geq n^2 . 
\end{equation} 
b)  The equality in (\ref{Qdim}) is achieved if and only 
if $V$ is a scalar multiple of a unitary matrix,
that is 
\begin{equation}\label{VU} 
V= \frac{1}{\sqrt{n}}\ G \,, \qquad G \in U(n) \,.
\end{equation} 
\end{propn} 
 
Thus, in the rank one case, $Q=n$ is the 
sharp lower bound. Moreover, 
for every $n \geq 2$, 
a unitary tensor space representation of
$TL_N(Q)$ of rank one exists for every $Q$ in the range 
$[n,+\infty)$ (see Theorem~\ref{ListUq} in Section~\ref{TLPR}).

\subsection{Higher rank case}

Let us now establish some estimates on the lower 
bound for $Q$ in the higher rank case.

\begin{thm}\label{Qrs}
If $T \in {M}_{n^2}$ has rank~$r \geq 1$ and
satisfies relations (T1)--(T4) with $Q>0$,
then the following inequalities hold: 
\begin{align}
\label{Qdim4} 
{} Q^4 & \geq   \frac{2n^2}{n^2+r} \,, \\[1mm] 
\label{Qdim2} 
{} Q  & \geq  \frac{n}{r}   \,.
\end{align}  
\end{thm} 

Inequality (\ref{Qdim4}) implies the following.

\begin{cor}\label{Q1} 
$Q=1$ is possible only for $r=n^2$ that is in the trivial
case $T= I_{n^2}$. 
\end{cor}

Next, we will find certain restrictions on the possible
values of $Q$ using the Jones--Wenzl orthogonal 
projector~\cite{Jo1,We1}. Recall that, for a generic 
value of $Q$, the algebra $TL_N(Q)$ with the 
anti--involution (\ref{TL0}) possesses 
a unique non--zero element $\SP_N$ such that
\begin{align}
\label{JW1}
{}& \SP_N \, \SP_N =\SP_N \,, \qquad \SP_N^* = \SP_N \,, \\
\label{JW2}
{}& \ST_k \, \SP_N = \SP_N \, \ST_k = 0 \,, 
 \qquad  \text{for}\ \  k= 1,\ldots, N-1 \,.
\end{align}
For the first three values of $N$, these
projectors are given by
\begin{eqnarray}\label{jw123}
 \SP_1 = 1 \,, \qquad
 \SP_2 = 1 - \frac{1}{Q} \, \ST_1 \,, \qquad
 \SP_3 = 1 -
 \frac{Q\,\bigl( \ST_1 + \ST_2  \bigr)}{Q^2-1} 
 + \frac{\bigl( \ST_1 \ST_2 + \ST_2 \ST_1 \bigr)}{Q^2-1}  \,.
\end{eqnarray}

Let $\tau_{n,r}$ be the unitary tensor space representation 
of $TL_N(Q)$ determined by a matrix $T \in M_{n^2}$ 
which has rank~$r \geq 1$ and satisfies (T1)--(T4). 
Denote $P_{n,r,N}= \tau_{n,r}(\SP_{N})$
and $d_{N}(n,r) = \tr_{1,\ldots, N} (P_{n,r,N})$, 
where $\tr_{1,\ldots, N}$ is the matrix trace in 
$M_{n^N}$.  

\vspace{1mm}{\small
\begin{ex}
For the projectors given in (\ref{jw123}), we have 
(cf. (\ref{tr123}) and (\ref{trt}))
\begin{equation}\label{drn}
 d_1(n,r)= n \,, \qquad
 d_{2}(n,r)= n^2-r \,, \qquad
 d_3(n,r)= n^3 -2rn \,.
\end{equation}
\end{ex}
}
\vspace{1mm} 

Note that relations (\ref{JW1}) imply that
$P_{n,r,N}$ is a positive semi--definite matrix.
Therefore, $d_{N}(n,r)$ must be non--negative.
But we see from (\ref{drn}) that $d_3(n,r)<0$
for $r>n^2/2$. This implies that every representation
$\tau_{n,r}$ of a rank $r>n^2/2$ can correspond 
only to the value $Q=1$ (for which $\SP_3$ is not defined).
By a similar analysis of values of $d_N(n,r)$
for $N \geq 3$, we establish the following statement.

\begin{thm}\label{QJW}
Suppose that $T \in {M}_{n^2}$ has rank 
$r > n^2/4$ and satisfies relations (T1)--(T4)
with $Q>0$. Then $Q$ in (T2) belongs to the
following discrete set of values: 
\begin{align}
\label{Qspec0} 
{}& \text{if}\quad 
 4\cos^2 \Bigl(\frac{\pi}{m+2}\Bigr) \leq \frac{n^2}{r}  
  < 4\cos^2\Bigl(\frac{\pi}{m+3} \Bigr)  ,
 \qquad m \in {\mathbb N} \,, \\
\label{Qspec} 
{}& \text{then}\quad
 Q \in J_m \equiv 
 \left\{ 2 \cos \Bigl(\frac{\pi}{k+2} \Bigr)\,, 
  \quad k=1,\ldots,m \right\} .
\end{align} 
\end{thm}

\begin{rem}
In the theory of von Neumann algebras, it is know \cite{Jo1,We1}
that the algebra $TL_{\infty}(Q)$ with the anti--involution 
(\ref{TL0}) admits a normalizable positive trace only if 
$Q \in J_\infty \cup [2,+\infty)$. 
The situation with unitary tensor space representations 
of $TL_N(Q)$ is somewhat different because the range 
of allowed values of $Q$ depends on the value of 
the parameter~$n^2/r$. In particular, if $r \leq n^2/4$, 
then the positive definiteness of the Jones--Wenzl
projector imposes no restrictions on~$Q$.
\end{rem}
\vspace{1mm}

Theorem~\ref{QJW} along with Corollary~\ref{Q1}
imply, in particular, the following.

\begin{cor}\label{Qcor}
a) There exists no unitary tensor space representation 
of $TL_N(Q)$ of rank 
$r \in \bigl(\frac{1}{2}n^2,n^2 \bigr)$. \\
b) Each unitary tensor space representation 
of $TL_N(Q)$ of rank
$r \in \bigl(\frac{1}{2}(3-\sqrt{5})n^2, 
\frac{1}{2}n^2 \bigr]$ 
corresponds to $Q = \sqrt{2}$. \\
c) Each unitary tensor space representation 
of $TL_N(Q)$ of rank
$r \in \bigl( \frac{1}{3} n^2 , 
\frac{1}{2}(3-\sqrt{5})n^2 \bigr]$ 
corresponds to either $Q = \sqrt{2}$ or 
$Q =\frac{1}{2}(1+\sqrt{5})$.
\end{cor}

{\small
\begin{ex}\label{Trn2}
For $n=r=2$, by Corollary~\ref{Qcor}, the only allowed 
value of $Q$ is $Q=\sqrt{2}$. In this case, a particular
solution to (T1)--(T4) is given by
\begin{equation}\label{ss1}
 T(\zeta) =   \frac{1}{\sqrt{2}}
 \begin{pmatrix}
  1 & 0 & 0 & i \zeta \\
  0 & 1 & i & 0 \\
  0 & -i & 1 & 0 \\
  -i\zeta^{-1} & 0 & 0 & 1
 \end{pmatrix} 
  \,, \qquad |\zeta|=1  \,. 
\end{equation}
The corresponding R--matrix appearing in (\ref{RT2})
was listed as $R_{H0.2}$ in \cite{Hie} among other
constant solutions to the Yang--Baxter equation  
in the $n=2$ case.
\end{ex}
}
\vspace*{1mm}

\begin{rem}
By Theorem~\ref{QJW}, for $n=r=3$, the only
allowed values of $Q$ are $\sqrt{2}$, 
$\frac{1}{2}(1+\sqrt{5})$, and~$\sqrt{3}$.
A solution corresponding to $Q=\sqrt{3}$
was constructed in~\cite{WX}.
\end{rem}

\subsection{Higher rank case, special cases}

First, let us remark that every unitary
tensor space representation of $TL_N(Q)$ of rank 
$r$ can be used to construct an infinite tower
of representations of the same rank for  
underlying spaces of higher dimensions.

\begin{propn}\label{Vsca}
Suppose that
$T =Q\,P_{\CT}$, $\CT \sim \{V_1,\ldots,V_r\}$
is a solution to (T1)--(T4).
Given $m \in \mathbb N$,
define $\tilde{\CT} \sim \{\tilde{V}_1,\ldots,\tilde{V}_r\}$,
where $\tilde{V}_k = \frac{1}{\sqrt{m}} I_m \otimes V_k$
for all~$k$ (or, alternatively, 
$\tilde{V}_k = \frac{1}{\sqrt{m}} V_k \otimes I_m$
for all~$k$).
Then $\tilde{T}= m\, Q\,P_{\tilde{\CT}}$ 
is a solution to (T1)--(T4).
\end{propn}

\begin{rem}
Let us stress that $\tilde{T}$ does not coincide with 
the Kronecker product of $T$ and the identity matrix. 
Indeed, $\tilde{T}$ has the same rank as $T$. Even if $T$ is 
the trivial solution, $\tilde{T}$ is non--trivial for $m>1$.
\end{rem}

Next, we will refine the estimates on the value
of $Q$ for representations where the
spanning vectors of the subspace $\CT$ have
certain specific properties.

Given an orthonormal basis $\{e_a\}_{a=1}^n$ of 
${\mathbb C}^n$, we will write 
$v \sim V$ if $V \in M_n$ is the matrix of
coefficients of a vector 
$v \in {\mathbb C}^n \otimes {\mathbb C}^n$, i.e. 
$v  = \sum_{a,b=1}^n V_{ab}\, e_a \otimes e_b$.
Relation (\ref{gVg}) implies that the following 
characteristics of a vector in 
${\mathbb C}^n \otimes {\mathbb C}^n$ 
are independent of the choice of a basis
of~${\mathbb C}^n$:\\
a) $v \sim V$ such that
$V$ is a symmetric or antisymmetric matrix;\\
b) $v \sim V$ such that
$V$ is a scalar multiple of a unitary matrix.

\begin{propn}\label{tausym}
Suppose that
$T =Q\,P_{\CT} \in {M}_{n^2}$ has rank $r$ and
satisfies (T1)--(T4)  and $\CT$ contains 
a non--zero vector $v \sim V$ such that matrix $V$ 
is symmetric or antisymmetric. Then  
 \begin{equation}\label{Qsym}   
  Q^2  \leq n^2 .
\end{equation}
\end{propn}

This statement along with Proposition~\ref{QVN}
implies, in particular, the following.

\begin{cor}
If $T =Q\,P_{\CT} \in {M}_{n^2}$ satisfies (T1)--(T4)
and $\CT \sim \{V\}$, where matrix $V$ 
is symmetric or antisymmetric, then $Q=n$.
\end{cor}

{\small 
\begin{ex}
See solutions listed in part a) of Theorem~\ref{ListUq}.
The corresponding matrices $V$ are (anti)symmetric
by the symmetry (\ref{cgcsym}) of the Clebsch--Gordan 
coefficients of the algebra $U(su_2)$.
\end{ex}
}

\begin{propn}\label{Guni}
Suppose that
$T =Q\,P_{\CT} \in {M}_{n^2}$ has rank $r$ and
satisfies (T1)--(T4)  and $\CT$ contains a non--zero vector $v \sim V$
such that $V$ is a scalar multiple of a unitary matrix.\\
a) Then  
\begin{equation}\label{Qguni}   
  Q^2 = \frac{n^2}{r} \,.
\end{equation}
b) If, in addition, $\frac{n^2}{4} < r < n^2$, then
either $r=\frac{n^2}{3}$ and $Q=\sqrt{3}$
or $r=\frac{n^2}{2}$ and $Q=\sqrt{2}$.
\end{propn}

\vspace*{1mm}
{\small 
\begin{ex} 
For $T(\zeta)$ given by (\ref{ss1}), we have
$T(\zeta)=\sqrt{2} \, P_{\CT}$, 
$\CT \sim \{V_1,V_2\}$, where 
\begin{equation}\label{ss2} 
V_1  = \frac{1}{\sqrt{2}} \biggl( \begin{matrix}
  i\zeta & 0 \\
  0 & 1
 \end{matrix} \biggr)
  \,, \qquad
   V_2  = \frac{1}{\sqrt{2}} \biggl( \begin{matrix}
  0 & i \\
  1 & 0
\end{matrix} \biggr) 
\,, \qquad |\zeta|=1\,.
\end{equation}
Both $V_1$ and $V_2$ are scalar multiples of unitary matrices. 
So, $Q=\sqrt{2}$ as required by Proposition~\ref{Guni}.
Applying the recipe of Proposition~\ref{Vsca}, we can
use $V_1$ and $V_2$ in order to construct a unitary
tensor space representation of $TL_N(Q)$ of rank two
corresponding to $Q=n/\sqrt{2}$ for any even~$n$.
\end{ex}
} 

\vspace*{1mm}
\begin{propn}\label{VVG1} 
Suppose that
$T =Q\,P_{\CT} \in {M}_{n^2}$ has rank $r$ and
satisfies (T1)--(T4)  and $\CT \sim \{V_1,\ldots,V_r\}$, 
where matrix $V_1$ is non--singular and either
$V_k = V_1 \, g_k$  for $k=2,\ldots,r$ or
$V_k = g_k \, V_1$  for $k=2,\ldots,r$,
where, in both cases, all $g_k$ are unitary.
Then the following inequality holds:
\begin{equation}\label{Qguni2}   
  Q^2 \geq \frac{n^2}{r} \,.
\end{equation}
\end{propn} 

\vspace{1mm}
\begin{rem}  
In the rank one case, by Proposition~\ref{QVN},
we have $Q \geq n$ and the lower bound is achieved only when 
the matrix of coefficients $V$ is a scalar multiple
of a unitary matrix. Therefore, in view of 
Proposition~\ref{Guni}, one could conjecture
that $Q \geq n/\sqrt{r}$ if $r \leq n^2/4$.
However, below (see Theorem~\ref{ListUqP}) 
we will construct a family of rank two solutions to (T1)--(T4) 
for $n=3$ for which $Q\in [2,+\infty)$.
This example refutes the conjecture  
since $2 < 3/\sqrt{2}$.
Thus, in the case $r \leq n^2/4$,
it remains an open problem to sharpen the 
estimate $Q \geq n/r$ established in Theorem~\ref{Qrs}.
\end{rem}

\section{Representations of rank one and two 
via $U_q(su_2)$}\label{UQ}
\subsection{$U_q(su_2)$ and Clebsch--Gordan decomposition}

Recall the definition of the universal enveloping Lie 
algebra $U(su_2)$ and its quantum deformation $U_q(su_2)$.
\begin{defn}
a) $U(su_2)$ is the unital *--algebra over $\mathbb C$
with generators $X^+$, $X^-$, $H$ 
and the following defining relations:
\begin{equation}\label{sss0}   
  H \, X^{\pm} - X^{\pm} H = \pm X^{\pm} \,, \qquad
    X^+ X^- - X^- X^+ = 2 H \,,
\end{equation}
\begin{equation}\label{ssst0}   
  H^* = H \,, \qquad (X^{\pm})^* = X^{\mp} .
\end{equation}
b)
$U_q(su_2)$, $q >0$, $q \neq 1$, is the unital *--algebra 
over $\mathbb C$ with generators $X^+$, $X^-$, $K$, $K^{-1}$ 
and the following defining relations:
\begin{equation}\label{sss}   
  K\, X^{\pm} = q^{\pm 1} X^{\pm} \, K \,, \qquad
    X^+  X^- - X^- X^+ = \frac{K^2-K^{-2}}{q-q^{-1}} \,,
\end{equation}
\begin{equation}\label{ssst}   
  K\,K^{-1} = K^{-1} K = 1 \,, \qquad
  (K^{\pm 1})^* = K^{\pm 1} \,, \qquad 
  (X^{\pm})^* = X^{\mp} \,.
\end{equation}
\end{defn}

For both algebras, the center is generated by the 
corresponding Casimir element: 
\begin{align}
 \label{casim1}  
 U(su_2):\qquad {}& 
 C_1= X^- X^+ + H(H+1) \,,\\
\label{casimq}     
 U_q(su_2):\qquad {}&
   C_q = X^- X^+ +
   \frac{(K-K^{-1})(qK-q^{-1}K^{-1})}{(q-q^{-1})^2} .
\end{align}

Both algebras become bialgebras if the comultiplication
is defined as follows:
\begin{align}
 \label{delef0}  
 U(su_2):\qquad {}& 
 \Delta(X^{\pm}) = X^{\pm} \otimes 1 + 
   1 \otimes X^{\pm} \,, \qquad
   \Delta(H) = H \otimes 1 + 1 \otimes H \,,\\
\label{delef}   
 U_q(su_2):\qquad {}&
  \Delta(X^{\pm}) = X^{\pm} \otimes K + 
   K^{-1} \otimes X^{\pm} \,, \qquad
   \Delta(K^{\pm 1}) = K^{\pm 1} \otimes K^{\pm 1} .
\end{align}

\begin{rem}
Setting formally $K^{\pm 1}=q^{\pm H}$ and considering 
the limit $q\to 1$, one recovers from the defining 
relations and comultiplication of $U_q(su_2)$ those 
of $U(su_2)$. Furthermore, the q--number, i.e.
a function ${\mathbb R}_+ \times {\mathbb C} \to {\mathbb C}$
defined as follows: $[t]_{q=1}=t$ and
\begin{equation}\label{qn}   
   [t]_q = \frac{q^t - q^{-t}}{q-q^{-1}}  
   		\qquad \text{for} \ q \neq 1 
\end{equation}
is continuous at $q=1$. For these reasons, one can
regard $U(su_2)$ as the limit of $U_q(su_2)$ as
$q \to 1$. In particular, the Clebsch--Gordan
coefficients of $U_q(su_2)$ are continuous functions
at $q=1$ and their limit as $q \to 1$ yields 
the Clebsch--Gordan coefficients of $U(su_2)$.
\end{rem}

An irreducible finite dimensional representation 
of  $U_q(su_2)$ is characterized by 
its highest weight $\Lambda$ which
is a non--negative integer.
Following the terminology used in physics, we will 
refer to $S=\frac{1}{2} \Lambda$ as spin. We denote by 
${\mathcal H}_S^q$ the irreducible $U_q(su_2)$--module
of dimension $n=2S+1$.
On ${\mathcal H}_S^q$,  
the Casimir element $C_q$ takes the value $[S]_q[S+1]_q$. 

The tensor square 
${\mathcal H}_S^q {\otimes} {\mathcal H}_S^q$ 
decomposes into a direct sum of irreducible modules, 
${\mathcal H}_S^q {\otimes} {\mathcal H}_S^q = 
\bigoplus_{J=0}^{2S} {\mathcal H}_J^q$.
Let $|m\rangle \in {\mathcal H}_S^q$  denote the
eigenvector of $K$ (or $H$ if $q=1$)
such that  $K|m\rangle =q^{m} |m\rangle$
(respectively, $H|m\rangle =m |m\rangle$).
Let $|J,m \rangle_q \in {\mathcal H}_J^q \subset 
{\mathcal H}_S^q {\otimes} {\mathcal H}_S^q$ 
denote the joint eigenvector of $\Delta(K)$ and
$\Delta(C_q)$ (or $\Delta(H)$ and $\Delta(C_1)$ if $q=1$), i.e.
\begin{equation}\label{jev1} 
\begin{aligned} 
 q \neq 1: &\ \ \
 \Delta(K)|J,m \rangle_q = q^{m} |J,m \rangle_q , \qquad  \ \  
 \Delta(C_q) |J,m \rangle_q = [J]_q[J+1]_q |J,m \rangle_q , \\
 q= 1: &\ \ \   
 \Delta(H)|J,m \rangle_{q=1} = m |J,m \rangle_{q=1} , \quad
  \Delta(C_1) |J,m \rangle_{q=1} = J(J+1) |J,m \rangle_{q=1} .
\end{aligned}   
\end{equation}
The sets of vectors
$\{ |J,m \rangle_q \}_{m=-J,\ldots,J}^{J=0,\ldots,2S}$ 
and 
$\{ |m_1 \rangle \otimes |m_2\rangle \}_{m_1,m_2=-S,\ldots,S}$
provide two orthonormal bases for 
${\mathcal H}_S^q {\otimes} {\mathcal H}_S^q$
related to each other as follows:
\begin{equation}\label{cgb}   
    |J,m \rangle_q = \sum_{m_1,m_2=-S}^S
    \bigl\{  S , S , m_1 ,  m_2 | J , m \bigr\}_q  \, 
    |m_1 \rangle \otimes |m_2\rangle \,,
\end{equation}
where $\{  S , S , m_1 ,  m_2 | J , m\}_q$ stands for 
the Clebsch--Gordan coefficient (see formulae
(\ref{cgc1})--(\ref{cgcsym}) in the Appendix).

 Let us identify the vector $e_a$ of the canonical 
basis $\{e_a\}_{a=1}^{2S+1}$ of ${\mathbb C}^{2S+1}$
with the vector $|S+1-a \rangle$ of~${\mathcal H}_S^q$.
Then, by (\ref{cgb}), the vector 
$|J,m \rangle_q \in  
  {\mathcal H}_S^q {\otimes} {\mathcal H}_S^q$ 
can be associated
with the matrix $V \in M_{2S+1}$ such that:
\begin{equation}\label{vjm}   
   V_{ab}   = \delta_{a+b+m,2S+2} \, 
   \bigl\{  S , S , S+1 -a ,  S+1-b \,| J,m \bigr\}_q  \,.
\end{equation}

{\small
\begin{ex}
For $S=1/2$, matrices
\begin{equation}\label{v12s12}
 V_1 = \frac{1}{\sqrt{q^2+1}} 
 \biggl( \begin{matrix}
  0 & q \\
  -1  & 0 
 \end{matrix} \biggr)  \,, \qquad
 V_2 = \frac{1}{\sqrt{q^2+1}} 
 \biggl( \begin{matrix}
  0 & 1 \\
  q  & 0 
 \end{matrix} \biggr)   
\end{equation} 
correspond, respectively to the vectors 
$|0,0 \rangle_q$ and $|1,0 \rangle_q$.
\end{ex}
}

Observe that $V_1$ and $V_2$ in (\ref{v12s12})
coincide with (\ref{vTq}) for $\zeta=-1$ and $\zeta=1$
(up to a sign and transposition, respectively).
Therefore they define unitary tensor space representations 
of $TL_N(Q)$. Motivated by this example,
we will look for other unitary tensor space representations 
of $TL_N(Q)$ determined in the same sense
by one or two joint eigenvectors of $\Delta(C_q)$
and $\Delta(K)$ (or their counterparts 
$\Delta(C_1)$ and $\Delta(H)$ if $q=1$).

\subsection{TL vectors and TL pairs}\label{TLPR}

Below, $\frac{1}{2}{\mathbb Z}_{\geq 0}$ stands
for the set $\{0,1/2,1,3/2,\ldots\}$.

\begin{defn}
a) Given $S \in \frac{1}{2}{\mathbb Z}_{\geq 0}$ and $q>0$, 
a vector 
$|J,m \rangle_q \in {\mathcal H}_S^q {\otimes} {\mathcal H}_S^q$ 
is called a TL vector if equations (T1)--(T4) admit a solution 
of the form $T=QP_q$, where  $Q>0$ and $P_q$ is the projection 
in ${\mathcal H}_S^q {\otimes} {\mathcal H}_S^q$ on the 
one dimensional subspace spanned by $|J,m \rangle_q $. \\[1mm]
b) Given $S \in \frac{1}{2}{\mathbb Z}_{\geq 0}$ and $q>0$, 
a pair of orthogonal vectors $|J_1,m_1 \rangle_q,\, 
|J_2,m_2 \rangle_q \in 
{\mathcal H}_S^q {\otimes} {\mathcal H}_S^q$ 
is called a TL pair if equations (T1)--(T4) admit a solution 
of the form $T=QP_q$, where $Q>0$ and $P_q$ is the projection in 
${\mathcal H}_S^q {\otimes} {\mathcal H}_S^q$
 onto the two dimensional subspace spanned by these vectors.
 \end{defn}
 
\begin{propn}\label{PVQ}
Given $S \in \frac{1}{2}{\mathbb Z}_{\geq 0}$,
a vector $|J,m \rangle_q $ (respectively, a pair of orthogonal
vectors $|J_1,m_1 \rangle_q$ and $|J_2,m_2 \rangle_q$) 
is a TL vector (respectively, a TL pair) either for all $q>0$ 
or only for a finite (possibly empty) set of values of~$q$.
\end{propn}

Recall that we are predominantly interested in solutions to
$\mathrm{(T1)}$--$\mathrm{(T4)}$ with varying
$Q$, i.e. solutions that depend on a parameter $q$ in such a way
that $Q=Q(q)$ is a non--constant function of~$q$.
Proposition~\ref{PVQ} simplifies considerably the task of
finding all such solutions if $P_q$ is  the projection
onto a subspace spanned by one or two joint
eigenvectors of $\Delta(C_q)$ and $\Delta(K)$.
Indeed, Proposition~\ref{PVQ} implies that 
it suffices to restrict consideration to the case $q=1$
and then verify which of the found solutions
remain solutions to  $\mathrm{(T1)}$--$\mathrm{(T4)}$
for all $q>0$. Such a strategy allows us to 
establish the following. 

\begin{thm}\label{ListUq}
a) For $q=1$,  the exhaustive list of TL vectors $|J,m \rangle_{q=1}$
 is given by:
\begin{eqnarray}
\label{j1a} &&
 |0,0 \rangle_{q=1}  \quad 
 \text{for all}\ \  S \in  \frac{1}{2}{\mathbb Z}_{\geq 0};\\
\label{j1b} &&
 |1,0 \rangle_{q=1}  \quad 
 \text{for}\ \  S =  \frac{1}{2};\\
 \label{j1c} &&
 |2,0 \rangle_{q=1}  \quad 
 \text{for}\ \  S =  \frac{3}{2}.
\end{eqnarray}
In all the three cases, the corresponding
value of $Q$ is $Q= 2S+1$.\\[1mm]
b) The exhaustive list of vectors $|J,m \rangle_{q} $ which
are TL vectors for all $q>0$ is given by:
\begin{eqnarray}
\label{jqa} &&
 |0,0 \rangle_{q}  \quad 
 \text{for all}\ \  S \in  \frac{1}{2}{\mathbb Z}_{\geq 0};\\
\label{jqb} &&
 |1,0 \rangle_{q}  \quad 
 \text{for}\ \  S =  \frac{1}{2}.
\end{eqnarray}
In both cases, the corresponding
value of $Q$ is $Q= [2S+1]_q$.
\end{thm}

\vspace*{0.5mm}
\begin{rem} 
It was observed long ago in the physical literature that 
(\ref{j1a}) and (\ref{jqa}) are TL vectors, see \cite{BaBa}
and \cite{BaKu}, respectively.
\end{rem}

\vspace*{0.5mm} 
{\small
\begin{ex}\label{ex32} 
For $S=\frac{3}{2}$, vector $|2,0 \rangle_{q} $ corresponds to 
the following matrix:
\begin{equation}\label{232}
 V = \frac{1}{\sqrt{q^{10}+q^6+q^4+1}}  
\Biggl( \begin{matrix}
  0 & 0 & 0 & q \\
  0 & 0 & q^4+q^2-1  & 0 \\
  0 & q^5-q^3-q & 0 & 0 \\
  - q^4 & 0 & 0 & 0
 \end{matrix} \, \Biggr)  \,.
\end{equation} 
It is interesting to remark that, for $S=3/2$, 
vector $|2,0 \rangle_{q}$ 
is a TL vector not only at $q=1$ but also at 
the points where $(q^4-1)^2=2q^4$, i.e. at
$q=(2 \pm \sqrt{3})^{\frac 14}$. 
These points correspond to $Q^2=12+18\sqrt{6}$.
\end{ex}
}

\vspace*{0.5mm} 
\begin{rem}
The fact that a vector $|J,m \rangle_q $ is not a TL vector
at $q=1$ does not exclude the possibility that it becomes
a TL vector at some other value of~$q$ (by Proposition~\ref{PVQ},
there can be only finite number of such values). 
Not aiming at finding all such cases, we give a particular
example below.
\end{rem}
 
\vspace*{0.5mm}  
{\small
\begin{ex} 
For $S= 1$, vector $|1,0 \rangle_{q} $ corresponds to 
the following matrix:
\begin{equation}\label{s1j1}
 V = \frac{1}{\sqrt{q^4+1}}  
\Biggl( \begin{matrix}
  0 & 0 & q \\
  0 & q^2 -1  & 0 \\
   -q & 0 & 0
 \end{matrix} \, \Biggr)  \,.
\end{equation} 
Clearly, the corresponding vector is not a TL vector
at $q=1$ since $V$ is singular at this point. However,
it becomes a TL vector at the points where $q^2-1=\pm q$, i.e.
at  $q=(\sqrt{5} \pm 1)/2$. These points correspond to $Q=3$.
\end{ex}
}
 
 \begin{thm}\label{ListUqP}
a) For $q=1$,  the exhaustive list of TL pairs $|J_1,m_1 \rangle_{q=1}$,
$|J_2,m_2 \rangle_{q=1} $ such that $J_1 \geq J_2$  is given by:
\begin{eqnarray}
\label{jj1a} &&
i)\ \ \ \, |1,m \rangle_{q=1}  , \ \ |1,-m \rangle_{q=1}; \qquad
ii)\ \ \, |1,m \rangle_{q=1}  , \ \ |1,0 \rangle_{q=1}; \\
\label{jj1b} &&
 iii)\ \, |2,m \rangle_{q=1}  , \ \ |1,-m \rangle_{q=1}; \qquad
iv)\ \ |2,m \rangle_{q=1}  , \ \ |1,0 \rangle_{q=1}; \\
\label{jj1c} &&
v)\ \ \ |2,m \rangle_{q=1}  , \ \ |2,-m \rangle_{q=1} .
\end{eqnarray}
\nopagebreak
In all of these cases, $S=1$, $m=\pm 1$, and the corresponding
value of $Q$ is $Q= 2$.\\[1mm]
b) The exhaustive list of pairs of vectors $|J_1,m_1 \rangle_q $,
$|J_2,m_2 \rangle_q $ such that $J_1 \geq J_2$ which
are TL vectors for all $q>0$ is given by:
\begin{eqnarray}
\label{jjqa} && 
 i) \ \ \, |2,1 \rangle_{q}  , \ \ |1,-1 \rangle_{q} ; \qquad
 ii) \ \ \, |2,-1 \rangle_{q}  , \ \ |1,1 \rangle_{q} . 
\end{eqnarray}
In both cases, $S=1$ and the corresponding
value of $Q$ is $Q= q^2+q^{-2} \equiv [2]_{q^2}$.
\end{thm}

\vspace*{0.5mm}
{\small
\begin{ex} 
For $m=1$, the TL pair ii) in (\ref{jj1a}) corresponds to
\begin{equation}\label{s1j11}
 V_1 = \frac{1}{\sqrt{2}} 
\Biggl( \begin{matrix}
  0 & 1 & 0 \\
  -1 & 0  & 0 \\
   0 & 0 & 0
 \end{matrix} \, \Biggr) , \qquad
 V_2 = \frac{1}{\sqrt{2}}
\Biggl( \begin{matrix}
  0 & 0 & 1 \\
  0 & 0  & 0 \\
   -1 & 0 & 0
 \end{matrix} \, \Biggr) .
\end{equation}  
The first TL pair in (\ref{jjqa}) corresponds to
\begin{equation}\label{s1j21}
 V_1 = \frac{1}{\sqrt{q^4+1}}  
\Biggl( \begin{matrix}
  0 & 1 & 0 \\
  q^2 & 0  & 0 \\
   0 & 0 & 0
 \end{matrix} \, \Biggr) , \qquad
 V_2 = \frac{1}{\sqrt{q^4+1}}  
\Biggl( \begin{matrix}
  0 & 0 & 0 \\
  0 & 0  & q^2 \\
   0 & -1 & 0
 \end{matrix} \, \Biggr) .
\end{equation} 
In \cite{By2}, we will construct a more general rank two
solution containing (\ref{s1j21}) as a particular case.
\end{ex}
}

\appendix
\section{Appendix}
\begin{proof}[\bf Proof of Proposition~\ref{TTQT}]
Part a). The case $T=0$ is trivial, so we assume that
$T\neq 0$. Take $N = 4$ and consider
$T_{12}=T \,{\otimes}\, I_n \,{\otimes}\, I_n$, 
$T_{23}=I_n \,{\otimes}\, T \,{\otimes}\,  I_n$, and
$T_{34}=I_n \,{\otimes}\, I_n \,{\otimes}\, T$.  
Using relations $\mathrm{(T3)}$--$\mathrm{(T4)}$
and taking into account that $T_{12}$ commutes with 
$T_{34}$, we obtain
\begin{equation}\label{1234}
\begin{aligned}
{}& T^2 \otimes T^m \equiv
T^2_{12} \,   T_{34}^m =  
 T_{12} \, T_{34} \, T_{12} \, T_{34}^{m-1}
 \stackrel{\scriptscriptstyle (T4)}{=} 
 T_{12} \, T_{34} \, T_{23} \, T_{34} \, T_{12} \, T_{34}^{m-1} \\
{}& 
  = T_{34} \, T_{12} \, T_{23} \, T_{12} \, T_{34}^m 
  \stackrel{\scriptscriptstyle (T3)}{=}
  T_{34} \, T_{12} \, T_{34}^m =
  T_{12} \, T_{34}^{m+1} \equiv T \otimes T^{m+1} ,
\end{aligned}
\end{equation}
for every positive integer $m$.
Whence, by taking the partial trace $\tr_{34}$, 
we obtain
\begin{equation}\label{tr1234}
  T^2 \, \tr\, (T^m) \, = T  \, \tr\, (T^{m+1}) \,.
\end{equation}

Since $T \neq 0$, we infer from (\ref{tr1234}) that
 $\tr(T^m)=0$ implies $\tr(T^{m+1})=0$. Therefore,
if $\tr T =0$, then $\tr(T^m)=0$ for every positive integer~$m$.
Thus, $T$ is a nilpotent. 
Suppose it is a nilpotent of order $k>2$. 
Then the r.h.s. of (\ref{1234}) vanishes for $m=k-1$
but the l.h.s. does not vanish (since $A \ot B \neq 0$ 
if $A, B \neq 0$). Thus, $T$ is a nilpotent
of order two.
 
If $\tr T \neq 0$,  then $T$ has at least one non--zero
eigenvalue and so it is not a nilpotent, i.e.
$T\neq 0$ and $T^2\neq 0$. 
In this case, (\ref{tr1234}) for $m=1$ 
implies that $\tr\,(T^2)\neq 0$. Whence it follows
that $T$ satisfies relation (T2) with
$Q= \tr (T^2)/\tr T \neq 0$.
 
Part b). By the part a), we know that relations 
(T3)--(T4) imply relation~(T2) irrespective of 
whether or not $T$ is Hermitian. Let us show that 
(T1)--(T3) imply~(T4). Consider 
the following Hermitian matrix: 
$H=T_{23}T_{12}T_{23}-T_{23}$.
Then we have
\begin{align*}
{}  \tr_{123} \bigl(H^2\bigr) 
{}  & \stackrel{\scriptscriptstyle (T2)}{=}  
 Q \,\tr_{123} \bigl( T_{23} T_{12} T_{23} T_{12}T_{23}
- 2 T_{23} T_{12}T_{23} + T_{23} \bigr)\\
{}  & \stackrel{\scriptscriptstyle (T3)}{=}  
 Q \,\tr_{123} \bigl(T_{23} - T_{23} T_{12} T_{23}\bigr) 
  = Q \,\tr \bigl(T_{23} - T_{12} T^2_{23}\bigr)  \\
{}  &  \stackrel{\scriptscriptstyle (T2)}{=}   
 Q \,\tr_{123} \bigl( T_{23} - Q\,   T_{12} T_{23} \bigr)
 \stackrel{\scriptscriptstyle (T2)}{=}
  Q \,\tr_{123} \bigl( T_{23} - T^2_{12} T_{23} \bigr) \\ 
 {}  &=  
 Q \,\tr_{123} \bigl( T_{23} -  T_{12} T_{23} T_{12} \bigr) 
 \stackrel{\scriptscriptstyle (T3)}{=}
  Q \,\tr_{123} \bigl( T_{23} -  T_{12}  \bigr) = 0 \,. 
\end{align*}
But, since $H$ is Hermitian, $\tr\,(H^2)=0$ 
implies that $H=0$, that is relation $\mathrm{(T4)}$ holds.
One can show in the same way that  
(T1)--(T2) along with (T4) imply~(T3).
\end{proof}

\begin{proof}[\bf Proof of Theorem~\ref{trTTb}]
Given an orthogonal projection $P \in M_{n^2}$ of 
rank $r$,
consider the following family of Hermitian matrices:
$H(\alpha)=P_{12} P_{23} P_{12} - \alpha\, P_{12}$,
$\alpha \in {\mathbb R}$.
If there exists $\alpha_0 >0$ such that
$H(\alpha_0)=0$, then $T=P/\sqrt{\alpha_0}$
satisfies relations (T1)--(T3)
and hence, by Proposition~\ref{TTQT}, relation (T4)
as well. We have
\begin{align*}
{} f(\alpha) \equiv \tr_{123}\bigl(H^2(\alpha)\bigr) &=  
\tr_{123} \bigl( P_{12} P_{23} P_{12} P_{23} P_{12}
- 2\alpha\, P_{12} P_{23} P_{12}  + \alpha^2 P_{12} \bigr)\\
{}& = \alpha^2 n r - 2\alpha\, \tr_{123} \bigl(P_{12} P_{23} \bigr)
 + \tr_{123} ( P_{12} P_{23} )^2   \,. 
\end{align*}
Since $H(\alpha)$ is Hermitian, $H^2(\alpha)$ 
is positive semi--definite and therefore
$f(\alpha) \geq 0$. This means
that the discriminant of the quadratic
polynomial $f(\alpha)$ is non--positive, i.e.
\begin{equation}\label{trP2} 
   \bigl( \tr_{123}\,(P_{12} P_{23}) \bigr)^2 \leq 
   n\,r \, \tr_{123}\, (P_{12} P_{23} )^2   
\end{equation}
for every projection~$P$.
Hence a necessary and sufficient condition 
for equation $f(\alpha)=0$ to have a solution 
is the condition that the discriminant 
of $f(\alpha)$ vanishes. Thus, $H(\alpha)$
vanishes for some $\alpha=\alpha_0$ if
and only if the inequality (\ref{trP2})
for a given $P$ becomes an equality.
For such $P$, $f(\alpha)$ acquires the
following form:
$f(\alpha)=  \bigl(\alpha n r - 
\tr_{123}\,(P_{12} P_{23}) \bigr)^2/(nr)$.
Hence $f(\alpha_0)=0$ for 
$\alpha_0= \tr_{123}\,(P_{12} P_{23})/(nr)$.

It remains to note that the condition 
$P_{12} P_{23} \neq 0$ guarantees that 
$\alpha_0 >0$. Indeed,
$\tr_{123}\,(P_{12} P_{23})= 
\tr_{123}\,(P^2_{12} P^2_{23}) =
\tr_{123}\,\bigl( (P_{12} P_{23})
  (P_{12} P_{23})^* \bigr) \geq 0$,
where the equality occurs only if $P_{12} P_{23}=0$.
\end{proof}

\begin{proof}[\bf Proof of Proposition~\ref{trTT}]

Relations (T1)--(T2) imply that
$T=QP$, where $P$ is an orthogonal projection
and $Q>0$. Therefore,
the hypotheses listed in b) imply the
same hypotheses in terms of~$P$, i.e.
the hypotheses of Theorem~\ref{trTTb}
(as noted at the end of the proof of
Theorem~\ref{trTTb}, the condition
$\tr_{123}\,(P_{12} P_{23}) \neq 0$
is equivalent to $P_{12} P_{23} \neq 0$).
Hence b) implies~a).

Verification that a) implies c) is the following:
\begin{align}
\label{ac1} {}&    
   \tr_{123}\, (T_{12} T_{23} )^2 = 
   \tr_{123}\, (T_{12} T_{23} T_{12} T_{23})
   \stackrel{\scriptscriptstyle (T3)}{=} 
   \tr_{123}\, (T_{12} T_{23} ) , \\
\label{ac2} {}&
   Q \tr_{123}\,(T_{12} T_{23})
   \stackrel{\scriptscriptstyle (T2)}{=} 
    \tr_{123}\,(T^2_{12} T_{23}) =
   \tr_{123}\,(T_{12} T_{23} T_{12})
   \stackrel{\scriptscriptstyle (T3)}{=} 
    \tr_{123}\,(T_{12} ) 
   \stackrel{\scriptscriptstyle (\ref{tr123})}{=} 
    Q \,r \, n .
\end{align}
Finally, it is clear that c) implies~b). 
\end{proof}

\begin{proof}[\bf Proof of Theorem~\ref{PROPTLW}]
Using that $E_{ab}E_{cd}=\delta_{bc}\,E_{ad}$ and
$\tr(E_{ab})=\delta_{ab}$, it is straightforward
to compute for $P_\CT$ given by (\ref{PiTau})
the following traces:
\begin{align}
\label{trPP}  
{}&  \tr_{123}\bigl( (P_{\CT})_{12} (P_\CT)_{23} \bigr) =
  \sum_{s,m=1}^r  
  \tr \Bigl(  V_s \bar{V}_m V^t_m V_s^* \Bigr) \,, \\
\label{trPP2}  
{}&  \tr_{123}\bigl( (P_{\CT})_{12} (P_\CT)_{23} \bigr)^2 =
  \sum_{s,s',m,m'=1}^r   
  \tr \bigl( V_s \bar{V}_{m'} V^t_m V_s^* 
  V_{s'} \bar{V}_m V^t_{m'} V_{s'}^* \bigr) \,.
\end{align}
 
Note that $ \tr_{123} \bigl(
(P_{\CT})_{12} (P_\CT)_{23} \bigr)
= \tr W_\CT W^*_\CT$. Therefore 
$(P_{\CT})_{12} (P_\CT)_{23} \neq 0$
iff $W_\CT \neq 0$ (cf. the proof of Theorem~\ref{trTTb}).

Consider the following partitioned matrix 
containing $r^2$ blocks of the size $n \,{\times}\, n$:
\begin{align}
\label{AV}  
{}&  A_\CT =  W_\CT W_\CT^* =
   \sum_{s,m=1}^r  E_{sm} \otimes
   \Bigl( \sum_{k=1}^r V_k \bar{V}_s V^t_m V_k^* \Bigr) \,.
\end{align} 
Observe that (\ref{trPP}) and (\ref{trPP2})
coincide with $\tr A_\CT$ and $\tr(A_\CT^2)$, respectively.
Therefore condition (\ref{trP}) for $P_\CT$ acquires
the following form: 
\begin{equation}\label{AAA}
 (\tr A_\CT )^2 = nr \, \tr(A_\CT^2) \,.
\end{equation} 

Suppose that $T=QP_\CT$ is a solution to (T1)--(T4).
Then, by Theorem~\ref{trTTb}, equality (\ref{AAA}) holds
and $W_\CT \neq 0$ so that $\beta^2 \equiv \tr(A_\CT)/nr >0$. 
Since $A_\CT$ is Hermitian, matrix 
$H=(A_\CT - \beta^2 I_{nr})^2$ is positive semi--definite.
But equality (\ref{AAA}) implies that $\tr H =0$. Whence $H=0$ 
and thus $A_\CT = \beta^2 I_{nr}$. (Another way to establish
this result is to invoke the Cauchy inequality for the eigenvalues
of $A_\CT$.) Therefore $Q W_\CT$ is unitary for $Q=1/\beta$. 
The converse implication is obvious:
if $Q W_\CT$ is unitary, then 
$A_\CT=Q^{-2} I_{nr}$ and therefore equality 
(\ref{AAA}) holds. Hence, by Theorem~\ref{trTTb}, 
$T=QP_\CT$ is a solution to (T1)--(T4).
\end{proof}
 
\begin{proof}[\bf Proof of Proposition~\ref{QVN}] 
a) Let $\lambda_1,{\ldots},\lambda_n >0$ 
be the singular values of~$V$. 
The first relation in (\ref{TLVVr})
means that $\sum_{a=1}^n \lambda^2_a=1$.
Since 
$|\det V|^2 = \det V^* V = \prod_{a=1}^n \lambda^2_a$,
the arithmetic--geometric mean inequality implies that 
$|\det V|^2 \leq n^{-n}$. Substituting
this inequality in the first formula in (\ref{muV}),
we obtain  the estimate~(\ref{Qdim}).
Let us remark that the same estimate follows from the 
second  formula in (\ref{muV}) if we apply
the Cauchy inequality:
$\tr\bigl( (V {V}^*)^{-1}\bigr)=
\sum_{a=1}^n \lambda_a^{-2} = 
(\sum_{a=1}^n \lambda_a^2)(\sum_{a=1}^n \lambda_a^{-2})
\geq 
\bigl(\sum_{a=1}^n \lambda_a\, \lambda_a^{-1}\bigr)^2 
= n^2$. \\
b) The preceding consideration shows that the
equality $Q=n$ is achieved when $|\det V|^2$ has its 
maximal possible value, $n^{-n}$, that is when the 
arithmetic--geometric mean inequality for $\lambda_a^2$ 
turns into an equality. This is possible only if all 
$\la_a$ are equal which in turn implies that $V^* V$ 
is a scalar multiple of the identity matrix and 
whence $V=\alpha G$, where $G$ is unitary 
(and $\alpha=n^{-\frac 12}$ by the first relation 
in~(\ref{TLVVr})). Clearly, any $V$ of such a form
satisfies the second relation in~(\ref{TLVVr}).
\end{proof}

\begin{proof}[\bf Proof of Theorem~\ref{Qrs}]

Consider a family of Hermitian matrices:
$J(\alpha)=T_{12}+T_{34} + \alpha\, T_{23}$,
$\alpha \in \mathbb R$.
For the trace in $M_{n^4}$, we have
(cf. (\ref{tr123}) and (\ref{trt})) 
$\tr_{1234}\,(T_{12}) = Q n^2 r$,
$\tr_{1234}\,(T_{12} T_{23}) = n^2 r$, etc.
Using these equalities, we find
\begin{align*}
 f(\alpha) \equiv
 \tr_{1234}\bigl(J(\alpha)^2\bigr)  {}& =   
 r n^2 Q^2 \alpha^2 + 4 r n^2 \alpha +
 2 r Q^2 (n^2+r)  \,.
\end{align*}
Since $J^2(\alpha)$ is positive semi--definite,
the quadratic polynomial $f(\alpha)$ must be 
non--negative. The condition that the discriminant
of $f(\alpha)$ is non--positive yields  
the first estimate in (\ref{Qdim2}).

In order to prove the second estimate in (\ref{Qdim2}),
recall that if $\{\sigma_a\}_{a=1}^n$ are 
the singular values of a matrix  $A \in {M}_{n}$,
then $||A||_p = \bigl(\sum_{a=1}^n \sigma_a^p \bigr)^{1/p}$,
$p \geq 1$, defines the Schatten $p$--norm of~$A$.
We will need the following properties of these norms
(see, e.g. \cite{Ber}, Proposition~9.2.3 and Proposition~9.3.6):
\begin{align}
\label{ss}
 ||A+B||_p \leq ||A||_p + ||B||_p  \,, \qquad\quad
 ||A\,B||_p \leq ||A||_{2p} \, ||B||_{2p} \,.
\end{align} 

Since $T$ is a solution to (T1)--(T4),
it has the form $T=Q P_\CT$. Let $W_\CT$ be
the corresponding $rn\,{\times}\,rn$ matrix 
given by~(\ref{WV}).
Using inequalities (\ref{ss}),
we obtain 
\begin{align*}
 ||W_{\CT}||_1 &= \Big|\Big|   \sum_{k,m=1}^r
  E_{km} \ot V_m \bar{V}_k \,\Big|\Big|_1 \leq
  \sum_{k,m=1}^r \Big|\Big| E_{km} \ot V_m \bar{V}_k 
  \,\Big|\Big|_1  \\
 {}&
 =
  \sum_{k,m=1}^r || V_m \bar{V}_k ||_1
    \leq \sum_{k,m=1}^r || V_m ||_2 \, ||V_k ||_2  
  = r^2 \,,
\end{align*}
where we have taken into account that
$||V_k||_2=1$ for all~$k$ by the normalization condition~(\ref{vv}). 
Thus, $||W_{\CT}||_1 \leq r^2$.
On the other hand,  by Theorem~\ref{PROPTLW},
$Q W_{\CT}$ is a unitary matrix. Therefore,
$Q ||W_{\CT}||_1 = r\,n$. Hence the estimate
(\ref{Qdim2}) follows.
\end{proof}

Before proving Theorem~\ref{QJW}, we will
prove an auxiliary statement.

\begin{lem}\label{dPRN}
For $d_{N}(n,r) = \tr_{1,\ldots, N} (P_{n,r,N})$,
the following recurrence relation holds:
\begin{equation}\label{dtrt} 
 d_{N+1}(n,r) = n \, d_{N}(n,r) - r \, d_{N-1}(n,r) \,.
\end{equation} 
For the initial values $d_0(n,r)=1$ and $d_1(n,r)=n$,
the solution to (\ref{dtrt}) is given by:
\begin{align}
\label{dna1} 
{}& \text{if  } r \neq \frac{n^2}{4}: \qquad
 d_{N} (n, r) = 
  r^{\frac{N}{2}}
 \, \frac{ \xi^{N+1} - \xi^{-N-1}}{\xi-\xi^{-1}} \,,  \qquad
 \text{where}\quad \xi+\xi^{-1} = \frac{n}{\sqrt{r}}  \,,\\
\label{dna2} 
{}& \text{if  } r = \frac{n^2}{4}: \qquad
 d_{N} (n, r) = 
  r^{\frac{N}{2}} \, (N+1) \,.
\end{align} 
\end{lem}

Note that $n^{-N} d_N(n,r)$ is a polynomial
in $r/n^2$ of degree~$\lfloor{N/2}\rfloor$.
 
\begin{proof}[\bf Proof of Lemma~\ref{dPRN}]
It is well known \cite{We1} that the idempotent $\SP_N$
defined by relations (\ref{JW1}) and (\ref{JW2})
satisfies the following recursion relation:
\begin{equation}\label{Pind} 
 \SP_{N+1} = \SP_N - \rho_N \, 
 \SP_N  \ST_{N} \, \SP_N \,,
\end{equation} 
where the scalar factor $\rho_N$  in turn
satisfies a recursion relation, namely
\begin{equation}\label{rhoind} 
  \rho_{N+1} = \bigl( Q - \rho_N \bigr)^{-1} \,,\qquad
  \rho_0=0.
\end{equation}
Now, taking into account that $\ST_N$ commutes with $\SP_{N-1}$, 
we verify relation~(\ref{dtrt}):
\begin{align*}
{}& n \, d_{N}(n,r) - d_{N+1}(n,r)  =
  n \, \tr_{1,\ldots,N} P_{n,r,N} -
  \tr_{1,\ldots,N+1} P_{n,r,N+1}  \\
{}& \stackrel{\scriptscriptstyle (\ref{Pind})}{=} 
\rho_{N} \tr_{1,\ldots,N+1} 
   \bigl( \tau_{n,r}( \SP_N \ST_{N} \SP_N) \bigr)
= \rho_{N} \tr_{1,\ldots,N+1} 
   \bigl( \tau_{n,r}( \SP_{N}  \ST_N) \bigr) \\
{}& \stackrel{\scriptscriptstyle (\ref{Pind})}{=}    
\rho_{N} \tr_{1,\ldots,N+1} 
   \bigl( \tau_{n,r}( \SP_{N-1}  \ST_N  - \rho_{N-1} \, 
   \SP_{N-1} \ST_{N-1}  \SP_{N-1}  \ST_N) \bigr) \\
{}& \stackrel{\scriptscriptstyle (T2)}{=} 
\rho_{N} \tr_{1,\ldots,N+1} 
   \bigl( \tau_{n,r}( \SP_{N-1}  \ST_N  - 
   Q^{-1} \rho_{N-1} \, 
   \ST_{N-1} \SP_{N-1} \ST^2_N \SP_{N-1} ) \bigr) \\
{}& =     
\rho_{N} \tr_{1,\ldots,N+1} 
   \bigl( \tau_{n,r}( \SP_{N-1} \ST_N  - 
   Q^{-1}\rho_{N-1} \, 
   \ST_{N}  \ST_{N-1} \ST_{N} \SP_{N-1}) \bigr) \\
{}& \stackrel{\scriptscriptstyle (T4)}{=}    
(1- Q^{-1} \rho_{N-1})\, \rho_{N} \tr_{1,\ldots,N+1} 
   \bigl( \tau_{n,r}( \SP_{N-1}  \ST_N ) \bigr) \\
{}& \stackrel{\scriptscriptstyle (\ref{rhoind})}{=}   
 Q^{-1} \tr_{1,\ldots,N+1} 
   \bigl(  P_{n,r,N-1} \otimes T \bigr) 
   = r \, d_{N-1}(n,r) \,. 
\end{align*}

The difference equation (\ref{dtrt})
can be rewritten in the form
$v_{N} = B \cdot v_{N-1}$, where 
\begin{equation} 
B = \begin{pmatrix}
    n & - r \\ 1 & 0 
   \end{pmatrix} , \qquad
v_N=    
 \begin{pmatrix}
    d_{N+1}(n,r) \\
    d_{N}(n,r) 
   \end{pmatrix} , \qquad
v_0=  
   \begin{pmatrix}
    n \\  1 
   \end{pmatrix} .
\end{equation}
Computing the $N$--th power of the matrix $B$, we obtain
solution (\ref{dna1})--(\ref{dna2}).
\end{proof}

\begin{proof}[\bf Proof of Theorem~\ref{QJW}]
First, note that, if $r<n^2/4$, then
$\xi$ in (\ref{dna1}) is positive and so
$d_N(n,r)$ is also positive.
For $r=n^2/4$, $d_N(n,r)$ is obviously positive
as well.

The solution to the recursion relation (\ref{rhoind})
is given by
\begin{align}
\label{rhosolv1} 
{}& \text{if   } Q \neq 2: \qquad
 \rho_{N} =  \frac{q^N - q^{-N}}{q^{N+1}-q^{-N-1}} \,,
  \qquad \text{where}  \quad q+q^{-1} = Q \geq 1 \,,\\
\label{rhosolv2} 
{}& \text{if   } Q = 2: \qquad
   \rho_{N} = \frac{N}{N+1}\,.
\end{align}  
If $q=e^{\pm \frac{i\pi}{k+2}}$, where 
$k \in {\mathbb N}$, then $\rho_{m}$
is finite for all $m \leq k$ but $\rho_{k+1}=\infty$.
By (\ref{Pind}), this implies that, for
$Q =2 \cos\bigl(\frac{\pi}{k+2}\bigr)$, 
$k \in {\mathbb N}$, the sequence of Jones--Wenzl 
projectors $\SP_N$ is defined only up to $N=k+1$. 

For $n^2 \geq r > n^2/4$, we have $|\xi|=1$
in (\ref{dna1}) and therefore 
$r^{-\frac{N}{2}} \, d_N(r,n)= 
 \frac{\sin(N+1)\gamma}{\sin\gamma}$,
where $1 \leq 2\cos \gamma= n/\sqrt{r} <2$.
Therefore, for 
$\gamma \in (\frac{\pi}{m+3},\frac{\pi}{m+2}]$,
$m \in {\mathbb N}$,
 we have $d_N(r,n) \geq 0$ for all 
$N \leq m+1$ and $d_{m+2}(r,n) < 0$. 
The latter inequality contradicts the 
positive semi--definiteness of $P_{n,r,N}$ and
therefore it requires the sequence of Jones--Wenzl 
projectors to terminate at some $N$ not exceeding $m+1$.
Which, by the preceding consideration, 
restricts the allowed values of $Q$ to the set
$\{ 2 \cos\bigl(\frac{\pi}{k+2}\bigr) ,\,\
 1 \leq k \leq m \}$.   
\end{proof}
 
\noindent
\begin{proof}[\bf Proof of Proposition~\ref{Vsca}]
If 
$\tilde{V}_k = \frac{1}{\sqrt{m}} I_m \otimes V_k$,
then $W_{\tilde{\CT}} = \frac{1}{m} I_m \otimes W_{\CT}$.
Therefore, if $Q W_{\CT}$ is unitary, so is 
$m Q W_{\tilde{\CT}}$. It remains to note that
$\tr\bigl(\tilde{V}_k \tilde{V}^*_p\bigr)=
\frac{1}{m}\tr\bigl(I_m \otimes V_k V^*_p\bigr)=
\frac{1}{m}\tr\bigl(I_m) \tr\bigl( V_k V^*_p\bigr)
=\delta_{kp}$.
The case $\tilde{V}_k = \frac{1}{\sqrt{m}}  V_k \otimes I_m$
is analogous.
\end{proof}

\noindent
\begin{proof}[\bf Proof of Proposition~\ref{tausym}]
Without a loss of generality, we can assume that
$v$ is of length one. Taking Remark~\ref{HGW} into account, 
we can choose an orthonormal spanning set for $\CT$ 
in such a way that $v$ is its first basis vector.
That is, $\CT \sim \{V_1,\ldots,V_r\}$, where
$V_1$ is (anti)symmetric.
In this case, the upper--left block of $Q W_{\CT}$ is 
$U \equiv Q V_1 \bar{V_1} = \pm Q V_1 V_1^*$.
Whence, by the normalization condition (\ref{vv}),
we have $|\tr U| = Q$. By Theorem~\ref{PROPTLW}, 
$QW_\CT$ is unitary. So, $U$ is a principal submatrix
of a unitary matrix and hence a contraction. 
Since $U$ (or $-U$) is positive semi--definite, it implies 
that all its eigenvalues lie between 0 and~1. Therefore,
$|\tr U| \leq n$, which imposes the restriction on~$Q$.
\end{proof}

\noindent
\begin{proof}[\bf Proof of Proposition~\ref{Guni}] 
a) Without a loss of generality, we can assume that
$v$ is of length one. Taking Remark~\ref{HGW} into account, 
we can choose an orthonormal spanning set for $\CT$ 
in such a way that $v$ is its first basis vector.
That is, $\CT \sim \{V_1,\ldots,V_r\}$, where
$V_1$ is a scalar multiple of a unitary matrix. 
By Theorem~\ref{PROPTLW}, $QW_\CT$ is unitary.
Therefore the upper--left block of 
$Q^2 W_{\CT} W_{\CT}^*$ equals to~$I_n$, that is 
\begin{equation}\label{VV11}   
 Q^{2} \sum_{k=1}^{r} V_k \bar{V}_1 V_1^t V_k^* = I_n \,.
\end{equation}
Since $V_1$ is a scalar multiple of a unitary matrix, 
we have $\bar{V}_1 V_1^t = \frac{1}{n} I_n$ , where
$1/n$ factor is due to the normalization~(\ref{vv}).
Substituting the latter relation in (\ref{VV11}),
taking trace, and using again the condition (\ref{vv}),
we obtain that $Q^2=n^2/r$.

b) If $n^2/4 < r < n^2$, then Theorem~\ref{QJW} applies.
Taking into account that $Q^2=n^2/r$, we conclude from
(\ref{Qspec0}) and (\ref{Qspec}) that
\begin{equation}\label{nrcos}
\frac{n^2}{r} = 4 \cos^2 \Bigl( \frac{\pi}{m+2}\Bigr) \,,
 \qquad m \in {\mathbb N} \,.
\end{equation} 
Note that the r.h.s. of (\ref{nrcos}) equals to 
$e^{\frac{2\pi i}{m+2}}+2+e^{-\frac{2\pi i}{m+2}}$,
which is a sum of three algebraic integers. Therefore,
the r.h.s. of (\ref{nrcos}) is itself an algebraic integer.
But the l.h.s. of (\ref{nrcos}) is a rational number.
It is well known (see, e.g. Theorem~206 in~\cite{HaWr})
that the only rational algebraic integers are ordinary
integers. Thus, $n^2/r$ is an integer from the
interval $(1,4)$. Hence $n^2/r = 2$ or $n^2/r = 3$.
\end{proof}

\begin{proof}[\bf Proof of Proposition~\ref{VVG1}]   
First, consider the case $V_k = V_1 g_k$.
By Theorem~\ref{PROPTLW}, $QW_\CT$ is unitary.
Therefore equality (\ref{VV11}) holds.
Substituting $V_k = V_1 g_k$ 
in (\ref{VV11}) and taking into
account that $V_1$ is invertible, we can rewrite (\ref{VV11})
as follows:
\begin{align}\label{veq4b} 
     \bar{V}_1 V_1^t + 
     \sum_{k=2}^r g_k  \bar{V}_1 V_1^t g_k^* = 
    Q^{-2} ( V_1^* V_1 )^{-1} \,. 
\end{align}
Since $g_k$ are unitary and $V_1$ satisfies (\ref{vv}),
the trace of the l.h.s. of (\ref{veq4b}) is equal to~$r$.
By the Cauchy inequality (cf. Proof of Proposition~\ref{QVN}),
the trace of the r.h.s. of (\ref{veq4b}) is greater
or equal to $n^2/Q^{2}$. Whence follows inequality~(\ref{Qguni2}).

Now consider the case $V_k = g_k V_1$. Since $T=Q P_\CT$
is a solution to (T1)--(T4), so is $T'''=Q P_{\CT'''}$
(cf. Corollary~\ref{BTH}). But $V'''_1 = V^*_1$ and
$V'''_k = V^*_k = V^*_1 g^*_k$. Therefore the
preceding considerations apply to $T'''$ yielding the same
inequality on~$Q$.
\end{proof}

\noindent
{\bf Clebsch--Gordan coefficients for $U_q(su_2)$ }

The Clebsch--Gordan coefficients appearing in (\ref{cgb})
are given by \cite{Kir}:
\begin{eqnarray}
\nonumber  
 &&  \!\!\!\!\!\!\!\!
   \bigl\{ S,S,k_1,k_2 \vert J,m \bigr\}_q =
 \delta_{k_1+k_2,m} \,
 q^{\frac{1}{2}(2S-J)(2S+J+1)+S(k_2-k_1)} 
 \, [J]! \left(\frac{[2J+1]}{[2S+J+1]!}\right)^{\frac 12}  \\
 \label{cgc1} 
 && \!\!\!   \times
\bigl(  [2S-J]! [S+k_1]![S-k_1]!
 [S+k_2]![S-k_2]! [J+k_1+k_2]![J-k_1-k_2]!  \bigr)^{\frac{1}{2}}    \\ 
 \nonumber   
 &&\!\!\!   \times
\sum_{l \geq 0} \frac{(-1)^l \, q^{-l(2S+J+1)} }%
 {[l]! [2S-J-l]! [S-k_1-l]! [S+k_2-l]! 
 [J-S+k_1+l]! [J-S-k_2+l]! } .
\end{eqnarray}
Here $[l]! \equiv \prod_{p=1}^l [p]_q$ if $l$ is a positive integer,
 $[0]!=1$, and $[l]!=\infty$ if $l$ is a negative integer.
 Due to the latter property,  the sum in (\ref{cgc1}) always terminates.
 For $q=1$,  expression (\ref{cgc1}) recovers the Clebsch--Gordan 
 coefficients for $U(su_2)$.
 
The Clebsch--Gordan coefficients satisfy the orthogonality relation:
\begin{equation}\label{cgcort}   
\sum_{k_1,k_2=-S}^S 
\bigl\{ S,S,k_1,k_2 \vert J_1,m_1 \bigr\}_q \,
\bigl\{ S,S,k_1,k_2 \vert J_2,m_2 \bigr\}_q
= \delta_{J_1,J_2} \,  \delta_{m_1,m_2} \,. 
\end{equation} 
 They also possess a number of symmetries including
the following one:
\begin{equation}\label{cgcsym}   
\bigl\{ S,S,m_2,m_1 \vert J,m \bigr\}_q  =
 (-1)^{2S-J} \,  \bigl\{ S,S,m_1,m_2 \vert J,m \bigr\}_{q^{-1}} \,.
\end{equation} 

Let $V$ be the matrix associated by (\ref{vjm}) to a vector 
$|J,m \rangle_{q=1} \in 
  {\mathcal H}_S^{q=1} {\otimes} {\mathcal H}_S^{q=1}$, $m \geq 0$.
Every row and column of $V$ has at most one non--zero entry.
Below we will need explicit expressions for the non--zero entries 
of the first few rows of~$V$. 
    
\begin{lem}\label{FpCG} 
Let $p \in \{0,1,2\}$. For all 
$S \in \frac{1}{2}{\mathbb Z}_{\geq 1}$
and $ 0 \leq m \leq \min (J,2S-p)$  we have 
\begin{equation}\label{cgcmp}   
    \bigl\{ S , S ,  S-p , m+p-S | J , m  \bigr\}_{q=1} 
    = f_p(J,m) \,  
  \left(\frac{(2J+1)(2S-p)!(2S-m-p)!(J+m)!}%
   {p! (m+p)! (2S-J)! (2S+J+1)! (J-m)!}\right)^{\frac 12} ,
\end{equation}
where
\begin{equation}\label{fmp1}    
  f_0(J,m) = 1 \,, \qquad 
   f_1(J,m) = J(J+1) - 2S(m+1) \,,
\end{equation} 
\begin{equation}\label{fmp2}    
 f_2(J,m)   =  f^2_1(J,m)  + 2(m+1-2S) \,  f_1(J,m) 
 + 2S(m+1)(m-2S) \,.
\end{equation}
\end{lem}
\begin{proof}
A direct verification with the help of formula~(\ref{cgc1})
for $q=1$.
\end{proof}

\begin{rem}\label{Mex}
The Clebsch--Gordan coefficient on the l.h.s.
of (\ref{cgcmp}) vanishes if $p=1$ and $m=2S$
and also if $p=2$ and $m=2S-1$ or $m=2S$.
In these cases, we have $f_1=0$ and $f_2=0$, 
respectively. Indeed, $m=2S$ implies $J=2S$ and 
hence $f_1=f_2=0$. Similarly, $m=2S-1$ implies 
either $J=2S$ or $J=2S-1$. In both cases, 
$|f_1|=2S$ and $f_2=0$.
\end{rem}

\begin{proof}[\bf Proof of Proposition~\ref{PVQ}]   
Let $P_q \in {M}_{(2S+1)^2}$
stand for the orthogonal projection in 
${\mathcal H}_S^q {\otimes} {\mathcal H}_S^q$
on  the one (or two) dimensional subspace  
spanned by the vector (respectively, the pair of vectors)
under consideration.
Consider the following function of~$q$:
\begin{equation}\label{fdP} 
  f(q)= \Bigl( \tr_{123}\,\bigl( (P_q)_{12} (P_q)_{23} \bigr) \Bigr)^2 -
   (2S+1)\,r \, \tr_{123}\, \bigl( (P_q)_{12} (P_q)_{23} \bigr)^2  \,,
\end{equation}
where $r=\mathrm{rank}(P_q)$.
Our proof of the proposition will be based on the following lemma:
\begin{lem}\label{FPQ}
Function  $f(q)$ is rational in~$q$.
\end{lem}
An immediate consequence of Lemma~\ref{FPQ} is that 
$f(q)$  either vanishes identically for all $q>0$
or it has only a finite (possibly empty) set of zeros on
the semi--axis $q>0$.
But, by Theorem~\ref{trTTb}, a solution to (T1)--(T4) of the form
$T=Q P_q$ exists only for such values of $q$ that $f(q)=0$.
\end{proof}

\begin{proof}[\bf Proof of Lemma~\ref{FPQ}]   
Let $V_1,\ldots,V_r$ be the matrices associated by formula (\ref{cgb}) 
to some set of vectors  $|J_1,m_1 \rangle_q, \ldots, |J_r,m_r \rangle_q$
in ${\mathcal H}_S^q {\otimes} {\mathcal H}_S^q$.
Let $P_q$ be the projection onto the subspace spanned by these vectors.
Define the corresponding function $f(q)$ by~(\ref{fdP}).
Using (\ref{trPP}) and (\ref{trPP2}) and taking into account that all entries 
of each $V_k$ are real, we get
\begin{equation}\label{pvv1} 
 f_1(q) \equiv  \tr_{123}\,\bigl( (P_q)_{12} (P_q)_{23} \bigr)  =
   \sum_{k_1,k_2=1}^r
   \tr \bigl( 
   V_{k_1} V^t_{k_1} V^t_{k_2}  V_{k_2} \bigr) \,,
\end{equation}
\begin{equation}\label{pvv2} 
 f_2(q) \equiv  \tr_{123}\, \bigl( (P_q)_{12} (P_q)_{23} \bigr)^2 =
   \sum_{k_1,k_2,k_3,k_4=1}^r
   \tr \bigl(
    V_{k_1} V^t_{k_2} V^t_{k_3} V_{k_4}
    V_{k_2} V^t_{k_1} V^t_{k_4} V_{k_3}  \bigr) \,.
\end{equation}
By (\ref{vjm}), matrices $V_{k} V_{k}$, $V_{k} V^t_{k}$ and
$V^t_{k} V_{k}$  are diagonal  for every $k$.
Observe that formula (\ref{cgc1}) along with the symmetry 
 (\ref{cgcsym}) imply that all non--zero entries of these matrices 
 contain no square roots of  $q$--factorials.  
 Therefore, all non--zero entries of  these matrices are
 rational functions in~$q$. Whence it is evident that $f_1(q)$
  is a rational function for all $r$ and $f_2(q)$
  is a rational function for $r=1$. 
  Now, consider $f_2(q)$
  for $r=2$. In (\ref{pvv2}), the contributions from 
  the terms with $k_1=k_2$
  or $k_3=k_4$ are rational functions because
  the matrices in (\ref{pvv2}) can be multiplied 
  in the following ways:
 $\tr\bigl((V_{k_1} V^t_{k_2}) (V^t_{k_3} (V_{k_4}
    (V_{k_2} V^t_{k_1}) V^t_{k_4}) V_{k_3})\bigr)$ and
 $\tr \bigl(    (V_{k_1} (V^t_{k_2} (V^t_{k_3} V_{k_4})
    V_{k_2}) V^t_{k_1}) (V^t_{k_4} V_{k_3})  \bigr)$, respectively.    
If   $k_1 \neq k_2$ and  $k_3 \neq k_4$, then   
either $k_1=k_4$ and $k_2=k_3$ or 
$k_1=k_3$ and $k_2=k_4$. The contributions from such terms
also yield rational functions  because, in these cases,
 the matrices  in (\ref{pvv2}) can be multiplied in 
 the following ways:
 $\tr\bigl( (V_{k_1} (V^t_{k_2} V^t_{k_3}) V_{k_4})
    (V_{k_2} (V^t_{k_1} V^t_{k_4}) V_{k_3}) \bigr)$ and
 $\tr \bigl(
   ( V_{k_3}  V_{k_1})  (V^t_{k_2} (V^t_{k_3} (V_{k_4}
    V_{k_2}) V^t_{k_1}) V^t_{k_4})   \bigr)$, respectively. 
Thus, we conclude that both $f_1(q)$ and $f_2(q)$ are  
rational functions in $q$ if $r=1,2$. Hence the
same holds for $f(q)=f^2_1(q)-(2S+1)r\,f_2(q)$.
\end{proof}

\begin{proof}[\bf Proof of Theorem~\ref{ListUq}]
Let $V$ be the matrix  associated to a vector
 $|J,m \rangle_q$ by formula (\ref{vjm}). Note that 
 $V \bar{V}$ is diagonal and it  is degenerate for all $m\neq 0$. 
 Thus, we have to restricts
 consideration to the case $m=0$. 

For $J=m=0$, the sum in (\ref{cgc1}) contains only one term 
($l=S-k_1=S+k_2$) and, therefore, by (\ref{vjm}), we have 
\begin{equation}\label{vabj0} 
 V_{ab} = \delta_{a+b,2S+2} \,
  \frac{(-1)^{a-1} q^{S+1-a}}{ \sqrt{ [2S+1]_q} } \,, \  \qquad
   a=1,\ldots,2S+1 \,.
 \end{equation}
 Thus, $ [2S+1]_q V \bar{V} = (-1)^{2S} I_{2S+1}$. Hence, 
 by Theorem~\ref{PROPTLW}, vector $|0,0 \rangle_q$
 is a TL vector for all $q>0$ and all
 $S \in  \frac{1}{2} \mathbb{Z}_{\geq 0}$.
 
 For $q>0$, $S=1/2$, $J=1$, $m=0$, the corresponding matrix $V$
 was given in~(\ref{v12s12}). It is easy to see,
 invoking Theorem~\ref{PROPTLW}, that it corresponds
 to a TL vector for all $q>0$. 
 
 In order to complete the proof, we have to consider
 vectors $|J,0 \rangle_q$  for $S \geq 1$ and $J \geq 1$.
 
 First, setting $q=1$, we will prove that, in this case, the list 
 of TL vectors given in  the part a) is exhaustive. 
 To this end, we will invoke Theorem~\ref{PROPTLW} again.
 For $q=1$, matrix $V$ is either
 symmetric or anti--symmetric due to~(\ref{cgcsym}).
 Therefore, the diagonal matrix $V \bar{V} = \pm V V^*$ 
 can be a scalar multiple of a unitary matrix only if all 
 the non--zero entries of $V$ have equal moduli. 
 In view of formula  (\ref{vjm}), the latter 
 condition is equivalent to the following requirement:
 \begin{equation}\label{cge1} 
  \bigl\{ S , S ,  S-p , p-S \,| J , 0  \bigr\}_{q=1}^2
  = \frac{1}{ 2S+1} \,,
 \end{equation}
 for $p=0,1,\ldots,2S$. The value on the r.h.s. of (\ref{cge1})
 is due to the normalization condition, $\tr (VV^*) =1$.
 
For $S \geq 1$, equality of the expressions 
on the l.h.s. of (\ref{cge1}) for $p=0,1,2$ is 
equivalent by Lemma~\ref{FpCG} to the following equalities:
 \begin{equation}\label{cge2} 
 \bigl( f_0(J,0) \, (2S)! \bigr)^2 =  
 \bigl( f_1(J,0) \, (2S-1)! \bigr)^2 = 
 \bigl( \frac{1}{2} f_2(J,0) \, (2S-2)! \bigr)^2 .
 \end{equation}
For $J>0$, formulae (\ref{fmp1}) imply that the first
 equality in (\ref{cge2}) holds iff 
  \begin{equation}\label{jjs1} 
   J(J+1) = 4S \,.
 \end{equation}
 Remarkably, under this condition, we have, by (\ref{fmp1})
and (\ref{fmp2}), $f_1(J,0)=2S$ and $f_2(J,0)=4S(1-2S)$
so that the second equality in (\ref{cge2}) holds identically 
thus imposing no further restrictions on $S$ and~$J$.
 
 Equality (\ref{jjs1}) provides a necessary condition for 
  $|J,0 \rangle_{q=1}$ to be a TL vector. In order to
  obtain another  necessary condition, we rewrite
 equality (\ref{cge1}) for $p=0$ 
 with the help of (\ref{cgcmp}) and (\ref{fmp1}) 
 in the following form:
\begin{equation}\label{jjs2} 
   \frac{ (2S-J)! (2S+J+1)! }{(2S)! (2S+1)!} = 2J+1\,.
 \end{equation}
For $J>2$, define $F(J) \equiv \bigl(J/(J-2)\bigr)^J$.
Observe that if $J$ and $S$ are related as in (\ref{jjs1})
and $J>2$, then we have the following estimate for 
the l.h.s. of~(\ref{jjs2}):
\begin{eqnarray*}
{}&&   \frac{ (2S-J)! (2S+J+1)! }{(2S)! (2S+1)!} = 
    \prod_{k=1}^J \frac{2S+1+k}{2S -J+k}
    =  \prod_{k=1}^J \Bigl( 1+ \frac{J+1}{2S -J+k} \Bigr) \\
{}&&    < \Bigl( 1+ \frac{J+1}{2S -J-1} \Bigr)^J 
=\Bigl( \frac{2S}{2S -J-1} \Bigr)^J 
 \stackrel{ (\ref{jjs1})}{=} F(J) \,.
 \end{eqnarray*}
 Note that $F(J)$ decreases monotonically as $J$ grows
 and we have $F(6) = (3/2)^6 \approx 11.4<2\cdot6+1$.
 Therefore, for all $J \geq 6$, the l.h.s. of (\ref{jjs2}) 
 is smaller then the r.h.s. So, it remains to check
 whether (\ref{jjs2})  holds for $1 \leq J \leq 5$ provided that
  $J$ and $S$ are related as in~(\ref{jjs1}).
  A~direct inspection shows that equalities (\ref{jjs1}) and  
(\ref{jjs2}) are not compatible for $J=3,4,5$ but they hold 
  for $J=1$ and $J=2$ if  the corresponding values of $S$ are 
 $S=1/2$ and $S=3/2$, respectively.
 The first of these cases,  $J=1$, $S=1/2$ was already 
 considered above. In the second case,  
it is easy to check that $|2,0 \rangle_{q=1}$ is 
indeed a TL vector if $S=3/2$ (cf. Example~\ref{ex32}). 
This completes the proof of the part~a) of the theorem.
 
 What the part~b) of the theorem is concerned,
 it was already explained at the beginning of this proof
 why the vectors listed in the part b) are TL vectors.
This list is exhaustive because the third case found 
for $q=1$, i.e. $|2,0 \rangle_{q}$ for $S=3/2$
is not a TL vector except for $q=1$ and two other values of~$q$ 
(cf. Example~\ref{ex32}).
\end{proof}

\begin{proof}[\bf Proof of Theorem~\ref{ListUqP}]
Part a).
Let $V_1$, $V_2$ be the matrices  associated 
by formula (\ref{vjm}) to a pair of orthogonal vectors
 $|J_1,m_1 \rangle_{q=1}$, $|J_2,m_2 \rangle_{q=1}$. 
Taking into account that  $V_1$, $V_2$ are real,
the unitarity condition (\ref{TLP}) is equivalent
to the following set of equations:
\begin{eqnarray}
\label{veq1} 
{}&   V_1 V_1 V_1^t V_1^t +
    V_2 V_1 V_1^t V_2^t = Q^{-2} I_n \,,\qquad 
    V_1 V_2 V_2^t V_1^t +
    V_2 V_2 V_2^t V_2^t = Q^{-2} I_n \,, & \\
\label{veq3} 
{}&   V_1 V_1 V_2^t V_1^t +
    V_2 V_1 V_2^t V_2^t = 0 \,, &
\end{eqnarray}
where $n=2S+1$ and $Q>0$.

Denote $A \equiv V_2 V_1$, $B \equiv V_1 V_1^t$,
$C \equiv V_2 V_2^t$. 
Recall that, by the symmetry (\ref{cgcsym}),
we have $V_1^t=(-1)^{2S-J_1} V_1$ and 
$V_1^t=(-1)^{2S-J_2} V_2$. Therefore,
equations (\ref{veq1})--(\ref{veq3}) can be 
rewritten in the following form:
\begin{eqnarray}
\label{aabb1} 
{}&   B^2= Q^{-2} I_n - A \, A^t \,,  \qquad 
   C^2 = Q^{-2} I_n - A^t  A \,, & \\
\label{baac1} 
{}&   B \, A + A \, C = 0 \,. &
\end{eqnarray}
Let the singular value decomposition of $A$ be
given by $A=O_1\Sigma\, O_2$, where $O_1$ and 
$O_2$ are orthogonal and $\Sigma \geq 0$ is diagonal
(note that $\Sigma \leq Q^{-1} I_n$).
Then, we infer from (\ref{aabb1}) that
\begin{align}
\label{aabb2} 
{}&   B^2 =  O_1 
 \bigl( Q^{-2}  I_n - \Sigma^2 \bigr)  O_1^t \,, \qquad 
   C^2 = O_2^t \bigl( Q^{-2}  I_n - \Sigma^2 \bigr)  O_2 \,.
\end{align}
These equations determine the singular values of 
$B^2$ and $C^2$. Taking into account that $B$ and $C$ are 
diagonal and positive semi--definite, we conclude that
\begin{align}
\label{aabb3} 
{}&   B =  \tilde{O}_1 
 \bigl( Q^{-2}  I_n - \Sigma^2 \bigr)^{\frac 12} \tilde{O}_1^t \,, 
 \qquad    C = \tilde{O}_2^t 
 \bigl( Q^{-2}  I_n - \Sigma^2 \bigr)^{\frac 12} \tilde{O}_2 \,,
\end{align}
where $\bigl(\ldots\bigr)^{\frac 12} \geq 0$ and $\tilde{O}_1$,
$\tilde{O}_2$ are orthogonal. 
Consistency of equations (\ref{aabb2}) and (\ref{aabb3}) implies that
$(Q^{-2}  I_n - \Sigma^2)$ commutes with
$\tilde{O}_1^t O_1$ and $O_2 \tilde{O}_2^t$.
Therefore, we get
\begin{align} 
\nonumber
{}& B\,A=  \tilde{O}_1 
 \bigl( Q^{-2}  I_n - \Sigma^2 \bigr)^{\frac 12} \tilde{O}_1^t
 O_1\Sigma\, O_2 =
 \tilde{O}_1 \tilde{O}_1^t O_1 \Sigma\,
 \bigl( Q^{-2}  I_n - \Sigma^2 \bigr)^{\frac 12}  O_2  \\
\label{baac} 
{}& =  O_1  \Sigma\,
 \bigl( Q^{-2}  I_n - \Sigma^2 \bigr)^{\frac 12}  O_2 =
  O_1  \Sigma\,
 \bigl( Q^{-2}  I_n - \Sigma^2 \bigr)^{\frac 12}  O_2 
 \tilde{O}_2^t \tilde{O}_2 \\
\nonumber 
{}& =  O_1  \Sigma\, O_2  \tilde{O}_2^t
 \bigl( Q^{-2}  I_n - \Sigma^2 \bigr)^{\frac 12} \tilde{O}_2 
 = A\, C \,.
\end{align}
But then (\ref{baac1}) holds only if
\begin{equation}\label{abc} 
B \, A = A \, C = 0 \,. 
\end{equation}
Equations (\ref{baac}) and (\ref{abc}) imply that
$\Sigma^2  \bigl( Q^{-2} I_n - \Sigma^2 \bigr) =0$. 
Hence the only non--zero eigenvalue of $\Sigma$ 
is~$Q^{-1}$. Recall that $B$ and $C$ are diagonal and 
positive semi--definite. Therefore, in view of 
(\ref{aabb3}), we have established the following.

\begin{lem}\label{BCQ}
Diagonal matrices $B \equiv V_1 V_1^t$ and 
$C \equiv V_2 V_2^t$ are singular, have equal ranks,
and all of their non--zero entries are equal to~$Q^{-1}$. 
\end{lem}

\begin{rem} If $B$ and $C$ were non--singular then so 
would be $V_1$, $V_2$ and hence also~$A$. But then 
equalities (\ref{abc}) could not hold.
\end{rem}

Recall that, in the rank one case (see the proof of
Theorem~\ref{ListUq}), all the entries of the 
diagonal matrix $VV^t$ are to be equal to~$Q^{-1}$.
Lemma~\ref{BCQ} shows that, in the rank two case, the 
situation is similar but somewhat more complicated 
because some of the diagonal entries of $B$ and $C$
can be equal to zero.

Note that if $|J_1,m_1 \rangle_{q=1}$, 
$|J_2,m_2 \rangle_{q=1}$ is a TL pair,
we must have $m_1 m_2 \leq 0$. Indeed,
suppose this is not so, for instances,
$m_1 <0$ and $m_2 < 0$. Then the first row
of $V_1$ and $V_2$ contains only zeroes and
hence $(B^2)_{11} = (AA^t)_{11} =0$ 
which contradicts relation~(\ref{aabb1}).
Thus, without a loss of generality, we can
assume that $m_1 \geq 0$ and $m_2 \leq 0$.
In this case, 
\begin{equation}\label{Bii} 
 B_{ii} =0  \qquad \text{ for all}\quad  i>2S+1-m_1 \,. 
\end{equation}
By Lemma~\ref{BCQ}, the remaining diagonal
entries of $B$ are equl to either zero or~$Q^{-1}$. 
In particular, the symmetry~(\ref{cgcsym}) implies that
$B_{2S+1-m_1,2S+1-m_1} = B_{11} = Q^{-1}$.

Let us show that, for $S \geq 9/2$, there exist no
$B$ and $C$ compatible with the requirements
imposed by Lemma~\ref{BCQ}. 
First, we note that, by (\ref{Bii}), $B$ has at most 
$2S+1-m_1$ non--zero entries and
hence, by Lemma~\ref{BCQ} and the normalization 
condition $\tr B =1$, the corresponding $Q$ must be 
an integer in the interval
\begin{equation}\label{nQn} 
  S + \frac{1}{2} \leq Q \leq 2S+1 -m_1 \,. 
\end{equation}
(The lower bound is imposed by Theorem~\ref{Qrs}).
Inequalities (\ref{nQn}) imply that $m_1 \leq S +1 /2$
and hence $2S+1 -m_1 \geq S +1 /2$.
Thus, for $S \geq 9/2$, matrix $B$ has at least five
entries which are not a priori zero. Moreover,
the first three of them, $B_{11}$, $B_{22}$, and $B_{33}$, 
are not related to each other by the symmetry~(\ref{cgcsym}).

By Lemma~\ref{FpCG}, $B_{11}$ is always non--zero.
Therefore, by Lemma~\ref{BCQ}, we have one of the
following cases: i) $B_{11}=Q^{-1}$, $B_{22}=B_{33}=0$;
ii) $B_{11}=B_{33}=Q^{-1}$, $B_{22}=0$; 
iii) $B_{11}=B_{22}=Q^{-1}$, $B_{33}=0$; 
iv) $B_{11}=B_{22}=B_{33}=Q^{-1}$.

Let now $f_1$ and $f_2$ stand for $f_1(J_1,m_1)$ and  
$f_2(J_1,m_1)$ defined in (\ref{fmp1}) and (\ref{fmp2}),
respectively. (Note that, for $S\geq 9/2$ and 
$m_1 \leq S+1/2$, we have $2S-m_1-p>0$ for $p=0,1,2$
so that the r.h.s. of (\ref{cgcmp}) is well defined.)
The case i) requires that $f_1=f_2=0$, which, by 
 (\ref{fmp2}), requires that $m_1 = 2S$. This
 is impossible since we have $m_1 \leq S +1 /2$.

 The case ii) requires that $B_{11}=B_{33}$
 which, by (\ref{cgcmp}), holds iff
\begin{equation}\label{jm00} 
  f_2^2= 4S(2S-1)(m_1 +1)(m_1 +2)(2S-m_1-1)(2S-m_1) .
\end{equation}
 On the other hand, by  (\ref{fmp1}), equality $B_{22}=0$, i.e.
 $f_1=0$, holds iff $J_1(J_1+1 )= 2S(m_1 +1)$.
In this case, (\ref{fmp2}) acquires the form 
$f_2= 2S(m_1+1)(m_1-2S)$.
Substituting the latter in (\ref{jm00}), we infer that  the case ii)
holds only if $J_1=2S-1$ and $m_1=2S-2$.
But the latter condition contradicts, for $S \geq 9/2$, the restriction
$m_1 \leq S +1 /2$.

In the cases iii) and iv), equality $B_{11}=B_{22}$ holds, 
as seen from (\ref{cgcmp}), only if 
\begin{equation}\label{jm22} 
f_1^2 = 2S(m_1+1)(2S-m_1) \,.
\end{equation}
Therefore, $f_2$ given by  (\ref{fmp2}) acquires
the following form:
\begin{equation}\label{jm33} 
f_2 = 2(m_1+1-2S) \, f_1\,.
\end{equation}
The case iii) requires that $f_2=0$ that is $m_1=2S-1$. But
this again  contradicts the restriction
$m_1 \leq S +1 /2$.

Finally, substituting  (\ref{jm22}) and (\ref{jm33}) into (\ref{jm00})
and taking into account that $m_1 \neq 2S-1, 2S$, we infer 
that, for $S \geq 9/2$, the case iv) holds only if $m_1=0$. 
In this case, relations (\ref{fmp1}) and (\ref{jm22}) imply
that either $J_1=0$ or $J_1(J_1+1)=4S$. 
Since $m_1=0$, we can assume that $m_2 \geq 0$ and,
repeating the same analysis for matrix $C$, we draw  
the conclusion that $m_2=0$ and either $J_2=0$ or $J_2(J_2+1)=4S$.
But the vectors determined by the matrices $V_1$ and $V_2$ 
must be orthogonal. Therefore, either $V_1$ or $V_2$
corresponds to $|0,0 \rangle_{q=1}$ and so 
it is a non--singular matrix, cf.~equation~(\ref{vabj0}). 
However, this is in contradiction with Lemma~\ref{BCQ}
which asserts that both $B$ and $C$ are singular.

Thus, we have proved that there exist no TL pairs 
$|J_1,m_1 \rangle_{q=1}$, $|J_2,m_2 \rangle_{q=1}$
for all \hbox{$S \geq 9/2$}.
A direct inspection shows that the only TL pairs 
for $S \leq 4$ are those 
listed in (\ref{jj1a})--(\ref{jj1c}). 
This completes the proof of the part~a).
 
Using explicit formulae (\ref{cgc1}) for the 
Clebsch--Gordan coefficients of $U_q(su_2)$, it is 
straightforward to check which of the pairs of vectors found
in the part~a) for $q=1$ remain TL pairs for other values
of~$q$.
\end{proof}

\vspace*{1mm}
\small{
{\bf Acknowledgements.} 
This work was supported by the grant MODFLAT of the European 
Research Council (ERC) and by the NCCR SwissMAP of the 
Swiss National Science Foundation,
and in part by the Russian Fund for Basic Research 
grants  14-01-00341 and 13-01-12405-ofi-m. 
}


\end{document}